\documentclass[12pt,cls,onecolumn]{IEEEtran}
\usepackage{graphicx,amsmath,amssymb,epsfig, amsfonts, cite, latexsym, cuted, multicol, multirow, subfigure, stfloats, array, tabularx}
\usepackage{subeqnarray}
\usepackage{color}
\usepackage{setspace}
\usepackage{anysize}

\begin{document}

\title{Opportunistic Network Decoupling With Virtual Full-Duplex Operation in Multi-Source Interfering Relay Networks}
\author{\large Won-Yong Shin, \emph{Senior Member}, \emph{IEEE}, Vien V. Mai, \\Bang Chul Jung \emph{Senior Member}, \emph{IEEE}, and Hyun Jong Yang, \emph{Member}, \emph{IEEE}
\\
\thanks{This work was supported by the Basic Science Research Program
through the National Research Foundation of Korea (NRF) funded by
the Ministry of Education (2014R1A1A2054577) and by the Ministry
of Science, ICT \& Future Planning (MSIP) (2015R1A2A1A15054248).
The material in this paper was presented in part at the IEEE
International Conference on Communications, Sydney, Australia,
June 2014~\cite{ONDICC:14}.}
\thanks{W.-Y. Shin is with the Department of Computer Science and
Engineering, Dankook University, Yongin 448-701, Republic of Korea
(E-mail: wyshin@dankook.ac.kr).}
\thanks{V.~V.~Mai was with Dankook University, Yongin 448-701, Republic of Korea. He is
now with KTH Royal Institute of Technology, Stockholm SE-100 44,
Sweden (E-mail: mienvanmaidt2@gmail.com).}
\thanks{B. C. Jung
is with the Department of Electronics Engineering, Chungnam
National University, Daejeon 305-764, Republic of Korea (E-mail:
bcjung@cnu.ac.kr).}
\thanks{H. J. Yang
(corresponding author) is with the School of Electrical and
Computer Engineering, UNIST, Ulsan 689-798, Republic of Korea
(E-mail: hjyang@unist.ac.kr).}
        %\thanks{The authors are with the School of Electrical Engineering and Computer Science, Korea Advanced Institute of Science and Technology (KAIST), Daejeon 305-701, Korea (E-mail: wyshin@stein.kaist.ac.kr; sychung@ee.kaist.ac.kr; yohlee@ee.kaist.ac.kr).}
} \maketitle

%\title{On the Power-Delay Trade-off in Wireless Ad Hoc Networks With Fading}
%\author{\authorblockN{Won-Yong Shin, Sae-Young Chung, and Yong H.
%Lee}\\
%\authorblockA{School of EECS, Korea Advanced Institute of Science and Technology, \\Daejeon, Korea\\
%Email: wyshin@stein.kaist.ac.kr; \{sychung, yohlee\}@ee.kaist.ac.kr}
%} \maketitle

\markboth{IEEE Transactions on Mobile Computing} {Shin {et al.}:
Opportunistic Network Decoupling With Virtual Full-Duplex
Operation in Multi-Source Interfering Relay Networks}

%%%%%%%%%%%%%%%%%%%%%%%%%%%%%%%%%%%%%%%%%%%%%%%%%%%%%%%%%%%%%%%%%%%%%%%%%%%%%%%%%%%%%%%%%%%%%%%%%%%%%%%%%%%%%%%%%%%%%%%%%%%%%%%%%%%%%%%%%%%%%%%%%%%%%

\newtheorem{definition}{Definition}%[section]
\newtheorem{theorem}{Theorem}%[section]
\newtheorem{lemma}{Lemma}%[section]
\newtheorem{example}{Example}
\newtheorem{corollary}{Corollary}%[section]
\newtheorem{proposition}{Proposition}%[section]
\newtheorem{conjecture}{Conjecture}%[section]
\newtheorem{remark}{Remark}%[section]

\newcommand{\red}[1]{{\textcolor[rgb]{1,0,0}{#1}}}
\newcommand{\blue}[1]{{\textcolor[rgb]{0,0,1}{#1}}}
\newcommand{\Vgreen}[1]{{\textcolor[rgb]{0,0.5,0}{#1}}}

\def \diag{\operatornamewithlimits{diag}}
\def \log{\operatorname{log}}
\def \rank{\operatorname{rank}}
\def \out{\operatorname{out}}
\def \exp{\operatorname{exp}}
\def \arg{\operatorname{arg}}
\def \E{\operatorname{E}}
\def \tr{\operatorname{tr}}
\def \SNR{\operatorname{SNR}}
\def \dB{\operatorname{dB}}
\def \ln{\operatorname{ln}}

\def \be {\begin{eqnarray}}
\def \ee {\end{eqnarray}}
\def \ben {\begin{eqnarray*}}
\def \een {\end{eqnarray*}}

\newcommand{\Pro}[1]{\mathrm{Pr}\left\{#1\right\}}
\newcommand{\LIF}[2]{\tilde{L}_{\pi_1(#1),#2}}
\newcommand{\TIL}[2]{L_{\pi_2(#1),#2}}
\newcommand{\TIF}[2]{T_{\pi_1(#1),#2}}
\newcommand{\KIF}[2]{T_{\pi_1(#1),\pi_2(#2)}}
\newcommand{\snr}{\textsf{snr}}
\newcommand{\sinr}{\textsf{sinr}}
\newcommand{\CanSB}{\mathcal{B}}
\newcommand{\CanSA}{\mathcal{A}}
\newcommand{\Norm}[1]{\left|{#1}\right|}
\newcommand{\PL}{\textsf{PL}}

\begin{abstract}
We introduce a new achievability scheme, termed {\em opportunistic
network decoupling (OND)}, operating in virtual full-duplex mode.
In the scheme, a novel relay scheduling strategy is utilized in
the {\em $K\times N\times K$ channel with interfering relays},
consisting of $K$ source--destination pairs and $N$ half-duplex
relays in-between them. A subset of relays using alternate
relaying is opportunistically selected in terms of producing the
minimum total interference level, thereby resulting in network
decoupling. As our main result, it is shown that under a certain
relay scaling condition, the OND protocol achieves $K$ degrees of
freedom even in the presence of interfering links among relays.
Numerical evaluation is also shown to validate the performance of
the proposed OND. Our protocol basically operates in a fully
distributed fashion along with local channel state information,
thereby resulting in a relatively easy implementation.
\end{abstract}


\newpage

\section{Introduction}

\subsection{Previous Work}

Interference between wireless links has been taken into account as
a critical problem in wireless communication systems. Recently,
interference alignment~(IA) was proposed for fundamentally solving
the interference problem when there are two communication
pairs~\cite{MaddahAliMotahariKhandani:08}. It was shown
in~\cite{Jafar_IA_original} that the IA scheme can achieve the
optimal degrees of freedom~(DoF), which is equal to $K/2$, in the
$K$-user interference channel with time-varying channel
coefficients. Since then, interference management schemes based on
IA have been further developed and analyzed in various wireless
network environments: multiple-input multiple-output (MIMO)
interference networks~\cite{Jafar_IA_distributed,Jafar_IA_MIMO}, X
networks~\cite{Jafar_Shamai}, and cellular
networks~\cite{Tse_IA,MotahariGharanMaddah-AliKhandani,JungShin,JPS:12:COM}.

On the one hand, following up on these successes for single-hop
networks, more recent and emerging work has studied multihop
networks with multiple source-destination (S--D) pairs. For the
2-user 2-hop network with 2 relays (referred to as the
$2\times2\times2$ interference channel), it was shown
in~\cite{GuoJafarWangJeonChung} that interference neutralization
combining with symbol extension achieves the optimal DoF. A more
challenging network model is to consider $K$-user two-hop
relay-aided interference channels, consisting of $K$
source-destination (S--D) pairs and $N$ helping relay nodes
located in the path between S--D pairs, so-called the $K\times
N\times K$ channel. Several achievability schemes have been known
for the network, but more detailed understanding is still in
progress. By applying the result from~\cite{RankovWittneben} to
the $K\times N\times K$ channel, one can show that $K/2$ DoF is
achieved by using orthogonalize-and-forward relaying, which
completely neutralizes interference at all destinations if $N$ is
greater than or equal to $K(K-1)+1$. Another achievable scheme,
called aligned network diagonalization, was introduced
in~\cite{ShomoronyAvestimehr} and was shown to achieve the optimal
DoF in the $K\times N\times K$ channel while tightening the
required number of relays. The scheme
in~\cite{ShomoronyAvestimehr} is based on the real interference
alignment framework~\cite{MotahariGharanMaddah-AliKhandani}.
In~\cite{GuoJafarWangJeonChung,ShomoronyAvestimehr}, however, the
system model under consideration assumes that there is no
interfering signal between relays and the relays are full-duplex.
Moreover, in~\cite{GouWangJafar}, the $2\times2\times2$
interference channel with full-duplex relays interfering with each
other was characterized and its DoF achievability was shown using
aligned interference neutralization.\footnote{The idea
in~\cite{GouWangJafar} was later extended to the 2-user 3-hop
network with 4 relays, i.e., the $2\times2\times2\times2$
interference channel~\cite{GouWangJafar_Asilomar}.}

On the other hand, there are lots of results on the usefulness of
fading in the literature, where one can obtain the multiuser
diversity gain in broadcast channels: opportunistic
scheduling~\cite{Knopp_Opp}, opportunistic
beamforming~\cite{Viswanath_Opp}, and random
beamforming~\cite{Hassibi_RBF}. Such opportunism can also be fully
utilized in multi-cell uplink or downlink networks by using an
opportunistic interference alignment
strategy~\cite{JPS:12:COM,H_Yang13_TWC,Codebook_OIA,ODIA_TSP14}.
Various scenarios exploiting the multiuser diversity gain have
been studied in cooperative networks by applying an opportunistic
two-hop relaying protocol~\cite{Poor_Opp} and an opportunistic
routing~\cite{Shin_Opp}, and in cognitive radio networks with
opportunistic scheduling~\cite{bcjung_CR, ShenFitz}. In addition,
recent results~\cite{Ergodic_Viswanath,JeonChung} have shown how
to utilize the opportunistic gain when there are a large number of
channel realizations. More specifically, to amplify signals and
cancel interference, the idea of opportunistically pairing
complementary channel instances has been studied in interference
networks~\cite{Ergodic_Viswanath} and multi-hop relay
networks~\cite{JeonChung}. In cognitive radio
environments~\cite{ZhangLiang}, opportunistic spectrum sharing was
introduced by allowing secondary users to share the radio spectrum
originally allocated to primary users via transmit adaptation in
space, time, or frequency.

\subsection{Main Contributions}

In this paper, we study the {\em $K\times N\times K$ channel with
interfering relays}, which can be taken into account as one of
practical multi-source interfering relay networks and be regarded
as a fundamentally different channel model from the conventional
$K\times N\times K$ channel in~\cite{ShomoronyAvestimehr}. Then,
we introduce an {\em opportunistic network decoupling (OND)}
protocol that achieves full DoF with comparatively easy
implementation under the channel model. This work focuses on the
$K\times N\times K$ channel with one additional assumption that
$N$ {\em half-duplex} (HD) relays interfere with each other, which
is a more feasible scenario. The scheme adopts the notion of the
multiuser diversity gain for performing interference management
over two hops. More precisely, in our scheme, a scheduling
strategy is presented in time-division duplexing (TDD) two-hop
environments with time-invariant channel coefficients, where a
subset of relays is opportunistically selected in terms of
producing the minimum total interference level. To improve the
spectral efficiency, the {\em alternate relaying} protocol
in~\cite{FanWangThompsonPoor,XueSandhu,Zhang08} is employed with a
modification, which eventually enables our system to operate in
{\em virtual full-duplex} mode. As our main result, it turns out
that in a high signal-to-noise ratio (SNR) regime, the OND
protocol asymptotically achieves the min-cut upper bound of $K$
DoF even in the presence of inter-relay interference and
half-duplex assumption, provided the number of relays, $N$, scales
faster than $\snr^{3K-2}$, which is the minimum number of relays
required to guarantee our achievability result.
%Furthermore, we examine that there exists a
%fundamental trade-off between the achievable DoF and the minimum
%required $N$ by comparing the two proposed schemes with and without
%alternate relaying.
Numerical evaluation also indicates that the OND protocol has
higher sum-rates than those of other relaying methods under
realistic network conditions (e.g., finite $N$ and SNR) since the
inter-relay interference is significantly reduced owing to the
opportunistic gain. For comparison, the OND scheme without
alternate relaying and the max-min SNR scheme are also shown as
baseline schemes. Note that our protocol basically operates with
local channel state information (CSI) at the transmitter and thus
is suitable for distributed/decentralized networks.

Our main contributions are fourfold as follows:

\begin{itemize}
\item In the $K\times N\times K$ channel with interfering relays,
we introduce a new achievability scheme, termed OND with virtual
full-duplex operation.

\item Under the channel model, we completely analyze the optimal
DoF, the required relay scaling condition, and the decaying rate
of the interference level, where the OND scheme is shown to
approach the min-cut upper bound on the DoF.

\item Our achievability result (i.e., the derived DoF and relay
scaling law) is validated via numerical evaluation.

\item We perform extensive computer simulations with other
baseline schemes.
\end{itemize}

\subsection{Organization}

The rest of this paper is organized as follows. In
Section~\ref{sec:2}, we describe the system and channel models. In
Section~\ref{sec:3}, the proposed OND scheme is specified and its
lower bound on the DoF is analyzed. Section~\ref{sec:4} shows an
upper bound on the DoF. Numerical results of the proposed OND
scheme are provided in Section~\ref{sec:5}. Finally, we summarize
the paper with some concluding remarks in Section~\ref{sec:6}.
%Numerical evaluation is shown in Section~\ref{SEC:Sim}.
%Finally, we
%summarize the paper with some concluding remark in
%Section~\ref{SEC:Conc}.

\subsection{Notations}

Throughout this paper, $\mathbb{C}$, $\mathbb{E}[\cdot]$, and
$\left\lceil{\cdot}\right\rceil$ indicate the field of complex
numbers, the statistical expectation, and the ceiling operation,
respectively. Unless otherwise stated, all logarithms are assumed
to be to the base 2.
% \textcolor{blue}{The function $f(x)$ defined by $f(x) = \omega(g(x))$ implies that $\lim_{x \rightarrow \infty} \frac{g(x)}{f(x)}=0$.}

Moreover, TABLE~\ref{Tab:Notations} summarizes the notations used
throughout this paper. Some notations will be more precisely
defined in the following sections, where we introduce our channel
model and achievability results.

\begin{table}[t]
    %\small
    \centering
    \caption{Summary of notations}
    \label{Tab:Notations}
\begin{tabular}{|c|c|}
  \hline
  Notation & Description \\
  \noalign{\hrule height 1.2pt}
  \hline
  $\mathcal{S}_k$ & $k$th source\\
  \hline
  $\mathcal{D}_k$ & $k$th destination\\
  \hline
  $\mathcal{R}_k$ & $k$th relay\\
  \hline
  $h_{ik}^{(1)}$ & channel coefficient from $\mathcal{S}_k$ to $\mathcal{R}_i$\\
  \hline
  $h_{ki}^{(2)}$ & channel coefficient from $\mathcal{R}_i$ to $\mathcal{D}_k$\\
  \hline
  $h_{ik}^{(r)}$ & channel coefficient between $\mathcal{R}_i$ and $\mathcal{R}_k$\\
  \hline
  $x_k^{(1)}(l)$ & $l$th transmitted symbol of $k$th source\\
  \hline
  \multirow{2}{*}{$\pi_s(k)~(s=1,2)$} & indices of two relays \\ & helping $k$th S--D pair\\
\hline
  $x_k^{(1)}(l)$ & $l$th transmit symbol of $\mathcal{S}_k$\\
  \hline
  $x_{\pi_s(k)}^{(2)}(l)$ & $l$th transmit symbol of
  $\mathcal{R}_{\pi_s(k)}$\\
  \hline
    $\mathbf{\Pi}_s~(s=1,2)$ & two selected relay sets\\
  \hline
      $\tilde{L}_{i,k}$ & scheduling metric in Step 1\\
  \hline
      $L_{i,k}^{\mathbf{\Pi}_2}$ & scheduling metric in Step 2\\
  \hline
        $\textsf{DoF}_\textsf{total}$ & total number of DoF\\
  \hline
   $\sinr_{\pi_s(k)}^{(1)}$ & SINR at $\mathcal{R}_{\pi_s(k)}$\\
  \hline
     $\sinr_{k,\pi_s(k)}^{(2)}$ & SINR at $\mathcal{D}_k$ (from $\mathcal{R}_{\pi_s(k)}$)\\
  \hline
\end{tabular}
\end{table}

%%%%%%%%%%%%%%%%%%%%%%%%%%%%%%%%%%%%%%%%%%%%%%%%%%%%%%%%%%%%%%%%%%%%%%%%%%%%%%%%%%%%%%%%%%%%%%%%%%%%%%%%%%%%%%%%%%%%%%%%%%%%%%%%%%%%%%%%%%%%%%%%%%

\section{System and Channel Models} \label{sec:2}

As one of two-hop cooperative scenarios, we consider the $K \times
N \times K$ channel model with interfering relays, which fits into
the case where each S--D pair is geographically far apart and/or
experiences strong shadowing (thus requiring the response to a
huge challenge for achieving the target spectral efficiency). In
the channel model, it is thus assumed that each source transmits
its own message to the corresponding destination only through one
of $N$ relays, and thus there is no direct path between an S--D
pair. Note that unlike the conventional $K \times N \times K$
channel, relay nodes are assumed to interfere with each other in
our model. There are $K$ S--D pairs, where each receiver is the
destination of exactly one source node and is interested only in
traffic demands of the source. As in the typical cooperative
relaying setup, $N$ relay nodes are located in the path between
S--D pairs so as to help to reduce path-loss attenuations.
%The
%example for $K=3$ and $N=10$ is shown in Fig.~\ref{fig:sm}.
%
%\begin{figure}[t]
%    \centerline{\includegraphics[width=0.32\textwidth]{Figs/sm.eps}}
%    \caption{The $K\times N\times K$ channel model with interfering relays where $K$=3 and $N=10$.
%    }
%    \label{fig:sm}
%\end{figure}

Suppose that each node is equipped with a single transmit antenna.
Each relay node is assumed to operate in half-duplex mode and to
fully decode, re-encode, and retransmit the source message i.e.,
decode-and-forward protocol is taken into account. We assume that
each node (either a source or a relay) has an average transmit
power constraint $P$. Unlike the work
in~\cite{GuoJafarWangJeonChung,ShomoronyAvestimehr}, $N$ relays
are assumed to interfere with each other.\footnote{If we can
cancel the interfering signals among multiple relays, then the
existing achievable scheme of the $K \times N \times K$ channel
can also be applied here.} To improve the spectral efficiency, the
alternate relaying protocol
in~\cite{FanWangThompsonPoor,XueSandhu,Zhang08} is employed with a
modification. With alternate relaying, each selected relay node
toggles between the transmit and listen modes for alternate time
slots of message transmission of the sources. If $N$ is
sufficiently large, then it is possible to exploit the channel
randomness for each hop and thus to obtain the opportunistic gain
in multiuser environments. In this work, we do not assume the use
of any sophisticated multiuser detection schemes at each receiver
(either a relay or a destination node), thereby resulting in an
easier implementation.

Now, let us turn to channel modeling. Let $\mathcal{S}_k$,
$\mathcal{D}_k$, and $\mathcal{R}_i$ denote the $k$th source, the
corresponding $k$th destination, and the $i$th relay node,
respectively, where $k\in\{1,\cdots,K\}$ and $i\in\{1,\cdots,N\}$.
The terms $h_{ik}^{(1)}, h_{ki}^{(2)} \in \mathbb{C}$ denote the
channel coefficients from $\mathcal{S}_k$ to $\mathcal{R}_i$ and
from $\mathcal{R}_i$ to $\mathcal{D}_k$, corresponding to the
first and second hops, respectively. The term $h_{in}^{(r)} \in
\mathbb{C}$ indicates the channel coefficient between two relays
$\mathcal{R}_i$ and $\mathcal{R}_n$. All the channels are assumed
to be Rayleigh, having zero-mean and unit variance, and to be
independent across different $i$, $k$, $n$, and hop index $r$. We
assume the block-fading model, i.e., the channels are constant
during one block (e.g., frame), consisting of one scheduling time
slot and $L$ data transmission time slots, and changes to a new
independent value for every block.

\begin{figure}[t]
    \centerline{\includegraphics[width=0.57\textwidth]{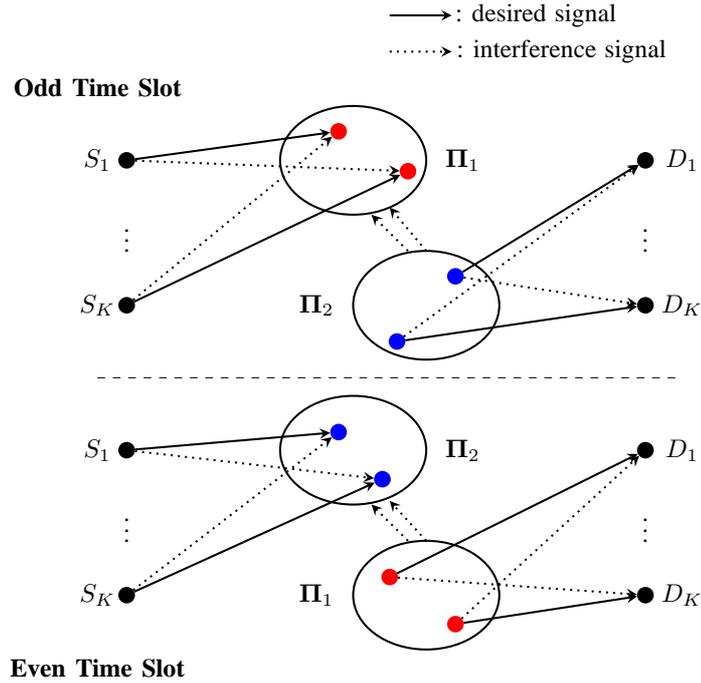}}
    \caption{The overall procedure of our OND scheme in the $K\times N\times K$ channel with interfering relays.
    }
    \label{fig:dig}
\end{figure}
%%%%%%%%%%%%%%%%%%%%%%%%%%%%%%%%%%%%%%%%%%%%%%%%%%%%%%%%%%%%%%%%%%%%%%%%%%%%%%%%%%%%%%%%%%%%%%%%%%%%%%%%%%%%%%%%%%%%%%%%%%%%%%%%%%%%%%%%%%%%%%%%%%%%%

\section{Achievability Results} \label{sec:3}

In this section, we describe the OND protocol, operating in
virtual full-duplex mode, in the $K\times N\times K$ channel with
interfering relays. Then, its performance is analyzed in terms of
achievable DoF along with a certain relay scaling condition. The
decaying rate of the interference level is also analyzed. In
addition, the OND protocol with no alternate relaying and its
achievability result are shown for comparison.

\subsection{OND in the $K\times N\times K$ Channel With Interfering Relays} \label{SEC:Protocol}
In this subsection, we introduce an OND protocol as the achievable
scheme to guarantee the optimal DoF of the $K \times N \times K$
channel with inter-relay interference, where $2K$ relay nodes
among $N$ candidates are opportunistically selected for data
forwarding in the sense of producing a sufficiently small amount
of interference level. The proposed scheme is basically performed
by utilizing the channel reciprocity of TDD systems.

Suppose that $\pi_1(k)$ and $\pi_2(k)$ denote the indices of two
relays communicating with the $k$th S--D pair for
$k\in\{1,\cdots,K\}$. In this case, without loss of generality,
assuming that the number of data transmission time slots, $L$, is
an odd number, the specific steps of each node during one block
are described as follows:

\begin{itemize}
\item Time slot 1: Sources $\mathcal{S}_1,\cdots,\mathcal{S}_K$
transmit their first encoded symbols
$x_1^{(1)}(1),\cdots,x_K^{(1)}(1)$, where $x_k^{(1)}(l)$
represents the $l$th transmit symbol of the $k$th source
node.\footnote{For notational convenience, we use scalar notation
instead of vector notation for each coding block from source
nodes, but the size of each symbol is assumed to be sufficiently
long to achieve the Shannon-theoretic channel capacity.} A set of
$K$ selected relay nodes,
$\mathbf{\Pi}_1=\{\pi_1(1),\cdots,\pi_1(K)\}$, operating in
receive mode at each odd time slot, listens to
$x_1^{(1)}(1),\cdots,x_K^{(1)}(1)$ (note that a relay selection
strategy will be specified later). Other $N-K$ relay nodes and
destinations $\mathcal{D}_1,\cdots,\mathcal{D}_K$ remain idle.

\item Time slot 2: The $K$ sources transmit their encoded symbols
$x_1^{(1)}(2),\cdots,x_K^{(1)}(2)$. The $K$ selected relays in the
set $\mathbf{\Pi}_1$ forward their first re-encoded symbols
$x_{\pi_1(1)}^{(2)}(1),\cdots,x_{\pi_1(K)}^{(2)}(1)$ to the
corresponding $K$ destinations. If the relays in $\mathbf{\Pi}_1$
successfully decode the corresponding symbols, then
$x_{\pi_1(k)}^{(2)}(1)$ is the same as $x_{k}^{(1)}(1)$. Another
set of $K$ selected relay nodes,
$\mathbf{\Pi}_2=\{\pi_2(1),\cdots,\pi_2(K)\}$, operating in
receive mode at each even time slot, listens to and decodes
$x_1^{(1)}(2),\cdots,x_K^{(1)}(2)$ while being interfered with by
$\mathcal{R}_{\pi_1(1)},\cdots,\mathcal{R}_{\pi_1(K)}$. The $K$
destinations receive and decode
$x_{\pi_1(1)}^{(2)}(1),\cdots,x_{\pi_1(K)}^{(2)}(1)$ from
$\mathcal{R}_{\pi_1(1)},$ $\cdots,\mathcal{R}_{\pi_1(K)}$. The
remaining $N-2K$ relays keep idle.

\item Time slot 3: The $K$ sources transmit their encoded symbols
$x_1^{(1)}(3),\cdots,x_K^{(1)}(3)$. The $K$ relays
$\pi_2(1),\cdots,\pi_2(K)$ forward their re-encoded symbols
$x_{\pi_2(1)}^{(2)}(2),\cdots,x_{\pi_2(K)}^{(2)}(2)$ to the
corresponding $K$ destinations. Another $K$ relays in
$\mathbf{\Pi}_1$ receive and decode $x_1^{(2)}(3),\cdots,$
$x_K^{(2)}(3)$ while being interfered with by
$\mathcal{R}_{\pi_2(1)},\cdots,\mathcal{R}_{\pi_2(K)}$. The $K$
destinations receive and decode
$x_{\pi_2(1)}^{(2)}(2),\cdots,x_{\pi_2(K)}^{(2)}(2)$ from
$\mathcal{R}_{\pi_2(1)},$ $\cdots,\mathcal{R}_{\pi_2(K)}$. The
remaining $N-2K$ relays keep idle.

\item The processes in time slots 2 and 3 are repeated to the
$(L-1)$th time slot.

\item Time slot $L$: The $K$ relays in $\mathbf{\Pi}_2$ forward
their re-encoded symbols $x_{\pi_2(1)}^{(2)}(L-1),\cdots,$
$x_{\pi_2(K)}^{(2)}(L-1)$ to the corresponding $K$ destinations.
The $K$ sources and the other $N-K$ relays remain idle.
\end{itemize}

At each odd time slot $l$ (i.e., $l=1,3,\cdots,L$), let us
consider the received signal at each selected relay for the first
hop and the received signal at each destination for the second
hop, respectively.

For the first hop (Phase 1), the received signal
$y_{\pi_1(i)}^{(1)}(l)\in\mathbb{C}$ at $\mathcal{R}_{\pi_1(i)}$
is given by
%\footnote{To simplify notations, the time index $l$ will be
%omitted if dropping $l$ does not cause any confusion.}
\begin{align}
y_{\pi_1(i)}^{(1)}(l) =&  \sum_{k=1}^K h_{\pi_1(i)k}^{(1)}
x_k^{(1)}(l) + \sum_{n=1}^K h_{\pi_1(i) \pi_2(n)}^{(r)}
x_{\pi_2(n)}^{(2)}(l-1) + z_{\pi_1(i)}^{(1)}(l), \label{EQ:y1}
\end{align}
where $x_k^{(1)}(l)$ and $x_{\pi_2(n)}^{(2)}(l-1)$ are the $l$th
transmit symbol of $\mathcal{S}_k$ and the $(l-1)$th transmit
symbol of $\mathcal{R}_{\pi_2(n)}$, respectively. As addressed
earlier, if relay $\mathcal{R}_{\pi_2(k)}$ successfully decodes
the received symbol, then it follows that
$x_{\pi_2(k)}^{(2)}(l-1)=x_k^{(1)}(l-1)$. The received signal
$y_{\pi_1(i)}^{(1)}(l)$ at $\mathcal{R}_{\pi_1(i)}$ is corrupted
by the independent and identically distributed (i.i.d.) and
circularly symmetric complex additive white Gaussian noise (AWGN)
$z_{\pi_1(i)}^{(1)}(l)$ having zero-mean and variance $N_0$. Note
that the second term in the right-hand side (RHS) of (\ref{EQ:y1})
indicates the inter-relay interference, which occurs when the $K$
relays in the set $\mathbf{\Pi}_1$, operating in receive mode,
listen to the sources, the relays are interfered with by the other
set $\mathbf{\Pi}_2$, operating in transmit mode. Note that when
$l=1$, relays have no symbols to transmit, and the second term in
the RHS of (\ref{EQ:y1}) becomes zero. Similarly when $l=L$,
sources do not transmit symbols, and the first term in the RHS of
(\ref{EQ:y1}) becomes zero.

For the second hop (Phase 2), assuming that the $K$ selected relay
nodes transmit their data packets simultaneously, the received
signal $y_k^{(2)}(l)\in\mathbb{C}$ at $\mathcal{D}_k$ is given by
\begin{align}
y_k^{(2)}(l)=\sum_{n=1}^K h_{k \pi_2(n)}^{(2)}
x_{\pi_2(n)}^{(2)}(l-1) + z_k^{(2)}(l), \label{EQ:y2}
\end{align}
where $z_k^{(2)}(l)$ is the i.i.d. AWGN having zero-mean and
variance $N_0$. We also note that when $l=1$, there are no signals
from relays.

Likewise, at each even time slot (i.e., $l=2,4,\cdots,L-1$), the
received signals at $\mathcal{R}_{\pi_2(i)}$ and $\mathcal{D}_k$
(i.e., the first and second hops) are given by
\begin{align}
y_{\pi_2(i)}^{(1)}(l)=& \sum_{k=1}^K h_{\pi_2(i)k}^{(1)}
x_k^{(1)}(l) + \sum_{n=1}^K h_{\pi_2(i) \pi_1(n)}^{(r)}
x_{\pi_1(n)}^{(2)}(l-1) + z_{\pi_2(i)}^{(1)}(l)
\nonumber %\label{EQ:y1_even}
\end{align}
and
\begin{align}
y_k^{(2)}(l)=\sum_{n=1}^K h_{k \pi_1(n)}^{(2)}
x_{\pi_1(n)}^{(2)}(l-1) + z_k^{(2)}(l), \nonumber %\label{EQ:y2_even}
\end{align}
respectively. The illustration of the aforementioned OND protocol
is geographically shown in Fig.~\ref{fig:dig} (two terms
$\tilde{L}_{\pi_1(k),k}$ and ${L}_{\pi_2(k),k}^{\mathbf{\Pi_2}}$
are specified later in the following relay selection steps).

Now, let us describe how to choose two types of relay sets,
$\mathbf{\Pi}_1$ and $\mathbf{\Pi}_2$, among $N$ relay nodes,
where $N$ is sufficiently large (the minimum $N$ required to
guarantee the DoF optimality will be analyzed in
Section~\ref{subsec:ana}).

\subsubsection{Step 1 (The First Relay Set Selection)}
\label{SubSub:1strelay}

Let us first focus on selecting the set
$\mathbf{\Pi}_1=\{\mathcal{R}_{\pi_1(1)},\cdots,\mathcal{R}_{\pi_1(K)}\}$,
operating in receive and transmit modes in odd and even time
slots, respectively. For every scheduling period, it is possible
for relay $\mathcal{R}_i$ to obtain all the channel coefficients
$h_{ik}^{(1)}$ and $h_{ki}^{(2)}$ by using a pilot signaling sent
from all of the source and destination nodes due to the channel
reciprocity before data transmission, where $i\in\{1,\cdots,N\}$
and $k\in\{1,\cdots,K\}$ (note that this is our local CSI
assumption). When $\mathcal{R}_i$ is assumed to serve the $k$th
S--D pair $(\mathcal{S}_k, \mathcal{D}_k)$, it then examines both
i) how much interference is received from the other sources and
ii) how much interference is generated by itself to the other
destinations, by computing the following scheduling metric
$\tilde{L}_{i,k}$:
\begin{equation} \label{eq:IL}
\tilde{L}_{i,k} = \sum_{m=1 \atop m\neq k}^{K}
\left(\left|h_{im}^{(1)}\right|^2 + \left|h_{mi}^{(2)}\right|^2
\right),
\end{equation}
where $i \in \{1, \ldots, N\}$ and $k \in  \{1, \ldots, K \}$. We
remark that the first term $\sum_{m=1, m\neq k}^{K}
\left|h_{im}^{(1)}\right|^2$ in (\ref{eq:IL}) denotes the sum of
interference power received at $\mathcal{R}_i$ for the first hop
(i.e., Phase 1). On the other hand, the second term $\sum_{m=1,
m\neq k}^{K} \left|h_{mi}^{(2)}\right|^2$ indicates the sum of
interference power generating at $\mathcal{R}_i$, which can be
interpreted as the \textit{leakage of interference} to the $K-1$
receivers expect for the corresponding destination, for the second
hop (i.e., Phase 2) under the same assumption.

%leakage of interference (LIF), $L_{i,k}$:
%\begin{align}
%L_{i,k}=\sum_k \left(\left|h_{ik}^{(1)}\right|^2+
%\left|h_{ki}^{(2)}\right|^2\right) \label{EQ:LIF}
%\end{align}
%for $k=1,\cdots,i-1,i+1,\cdots,K$, which represents the total sum of
%$2(K-1)$ squared channel links.

Suppose that a short duration CTS (Clear to Send) message is
transmitted by the destination who finds its desired relay node
(or the master destination). Then according to the computed
metrics $\tilde{L}_{i,k}$ in (\ref{eq:IL}), a timer-based method
can be used for the relay selection similarly as
in~\cite{BletsasKhistiReedLippman}.\footnote{The reception of a
CTS message, which is transmitted from a certain destination,
triggers the initial timing process at each relay. Therefore, no
explicit timing synchronization protocol is required among the
relays~\cite{BletsasKhistiReedLippman,BletsasShinWin}. Moreover,
it is worth noting that the overhead of relay selection is a small
fraction of one transmission block with small collision
probability~\cite{BletsasKhistiReedLippman}. Since our relay
selection procedure is performed sequentially over all the S-D
pairs and the already selected relays for a certain S-D pair are
not allowed to take part in the selection process for another S-D
pair, the collision probability is thus at most $2K$ times that of
the single S-D pair case~\cite{BletsasKhistiReedLippman}.} Note
that the method based on the timer is considerably suitable in
distributed systems in the sense that information exchange among
all the relay nodes can be minimized. At the beginning of every
scheduling period, the relay $\mathcal{R}_i$ computes the set of
$K$ scheduling metrics,
$\{\tilde{L}_{i,1},\cdots,\tilde{L}_{i,K}\}$, and then starts its
own timer with $K$ initial values, which can be set to be
proportional to the $K$ metrics.\footnote{To avoid a situation
such that a malicious relay deliberately sets its timer to a
smaller value so as to win the chance, prior to the relay
selection process, a secret key may be shared among legitimate
nodes including relays. If a malicious relay who did not share the
key wants to participate in communication, then one can neglect
his/her message (e.g., RTS (Request to Send) message).} Thus,
there exist $NK$ metrics over the whole relay nodes, and we need
to compare them so as to determine who will be selected. The timer
of the relay $\mathcal{R}_{\pi_1(\hat{k})}$ with the least one
$\tilde{L}_{\pi_1(\hat{k}),\hat{k}}$ among $NK$ metrics will
expire first, where $\pi_1(\hat{k})\in\{1,\cdots,N\}$ and
$\hat{k}\in\{1,\cdots,K\}$. The relay then transmits a short
duration RTS message, signaling its presence, to the other $N-1$
relays, where each RTS message is composed of $\left\lceil{\log_2
K}\right\rceil$ bits to indicate which S--D pair the relay wants
to serve. Thereafter, the relay $\mathcal{R}_{\pi_1(\hat{k})}$ is
first selected to forward the $\hat{k}$th S--D pair's packet. All
the other relays are in listen mode while waiting for their timer
to be set to zero (i.e., to expire). At the stage of deciding who
will send the second RTS message, it is assumed that the other
relays are not allowed to communicate with the $\hat{k}$th S--D
pair, and thus the associated metrics
$\{\tilde{L}_{1,\hat{k}},\cdots,\tilde{L}_{\pi_1(\hat{k})-1,\hat{k}},\tilde{L}_{\pi_1(\hat{k})+1,\hat{k}},\cdots,\tilde{L}_{N,\hat{k}}\}$
are discarded with respect to timer operation. If another relay
has an opportunity to send the second RTS message of
$\left\lceil{\log_2 (K-1)}\right\rceil$ bits in order to declare
its presence, then it is selected to communicate with the
corresponding S--D pair. When such $K$ RTS messages, consisting of
at most $K\left\lceil{\log_2 K}\right\rceil$ bits, are sent out in
consecutive order, i.e., the set of $K$ relays,
$\mathbf{\Pi}_1=\{\mathcal{R}_{\pi_1(1)},\cdots,\mathcal{R}_{\pi_1(K)}\}$,
is chosen, the timer-based algorithm for the first relay set
selection terminates, yielding no RTS collision with high
probability. We remark that when $K=1$ (i.e., the single S--D pair
case), $K$ relay nodes are {\em arbitrarily} chosen as the first
relay set $\mathbf{\Pi}_1$ since there is no interference in this
step.

\subsubsection{Step 2 (The Second Relay Set Selection)}

Now let us turn to choosing the set of $K$ relay nodes (among
$N-K$ candidates), $\mathbf{\Pi}_2=\{\pi_2(1),\cdots,$
$\pi_2(K)\}$, operating in receive and transmit modes in even and
odd time slots, respectively. Using $K$ RTS messages broadcasted
from the $K$ relay nodes in the set $\mathbf{\Pi}_1$, it is
possible for relay node $\mathcal{R}_i\in \{1,\cdots,N\} \setminus
\mathbf{\Pi_1}$ to compute the sum of inter-relay interference
power generated from the relays in $\mathbf{\Pi}_1$, denoted by
$\sum_{k=1}^{K} \left|h_{i \pi_1(k)}^{(r)}\right|^2$.
%\begin{equation} \label{eq:IRI}
%\sum_{k=1}^{K} \left|h_{i \pi_1(k)}^{(r)}\right|^2.
%\end{equation}
When $\mathcal{R}_i$ is again assumed to serve the $k$th S--D pair
$(\mathcal{S}_k, \mathcal{D}_k)$, it examines both i) how much
interference is received from the undesired sources and the
selected relays in the set $\mathbf{\Pi}_1$ for the first hop and
ii) how much interference is generated by itself to the other
destinations by computing the following metric
$L_{i,k}^{\mathbf{\Pi_2}}$, termed \textit{total interference
level~(TIL)}:
\begin{align}  %\label{eq:TIL}
L_{i,k}^{\mathbf{\Pi}_2} &= \tilde{L}_{i,k} + \sum_{k=1}^{K}
\left|h_{i
\pi_1(k)}^{(r)}\right|^2 \nonumber\\
&= \sum_{m=1 \atop m\neq k}^{K} \left(\left|h_{im}^{(1)}\right|^2
+ \left|h_{mi}^{(2)}\right|^2 \right) + \sum_{k=1}^{K} \left|h_{i
\pi_1(k)}^{(r)}\right|^2, \label{eq:TIL}
\end{align}
where $i \in \{1, \ldots, N\}$ and $k \in  \{1, \ldots, K \}$. We
note that Steps 1 and 2 cannot be exchangeable due to the fact
that the inter-relay interference term $\sum_{k=1}^{K} \left|h_{i
    \pi_1(k)}^{(r)}\right|^2$ is measured after determining the
    first relay set $\mathbf{\Pi}_1$. If the relay set selection
    order is switched, then the metric TIL in (\ref{eq:TIL}) will not be
    available.

According to the computed TIL $L_{i,k}^{\mathbf{\Pi_2}}$, we also
apply the timer-based method used in Step 1 for the second relay
set selection. The relay $\mathcal{R}_i\in \{1,\cdots,N\}
\setminus \mathbf{\Pi_1}$ computes the set of $K$ TILs,
$\{L_{i,1}^{\mathbf{\Pi_2}},\cdots,L_{i,K}^{\mathbf{\Pi_2}}\}$,
and then starts its timer with $K$ initial values, proportional to
the $K$ TILs. Thus, we need to compare $(N-K)K$ TIL metrics over
the relay nodes in the set $\{1,\cdots,N\} \setminus
\mathbf{\Pi_1}$ in order to determine who will be selected as the
second relay set. The rest of the relay set selection protocol
(i.e., RTS message exchange among relay nodes) almost follows the
same line as that of Step 1. The timer-based algorithm for the
second relay set selection terminates when $K$ RTS messages are
sent out in consecutive order. Then, $K$ relay nodes having a
sufficiently small amount of TIL $L_{i,k}^{\mathbf{\Pi_2}}$ are
selected as the second relay set $\mathbf{\Pi}_2$.

\begin{remark}
Owing to the channel reciprocity of TDD systems, the sum of
inter-relay interference power received at any relay $R_i \in
\mathbf{\Pi_1}$, $\sum_{k=1}^{K} \left|h_{i
\pi_2(k)}^{(r)}\right|^2$, also turns out to be sufficiently small
when $N$ is large. That is, it is also guaranteed that $K$
selected relays in the set $\mathbf{\Pi}_1$ have a sufficiently
small amount of TIL.
\end{remark}

\begin{remark}
The overhead of each scheduling time slot (i.e., the total number
of bits required for exchanging RTS messages among the relay
nodes) can be made arbitrarily small, compared to one transmission
block. From the fact that $K$ RTS messages, consisting of at most
$K\left\lceil{\log_2 K}\right\rceil$ bits, are sent out in each
relay set selection step, only $2K\left\lceil{\log_2
K}\right\rceil$ bit transmission could suffice.
\end{remark}

\subsubsection{Step 3 (Data Transmission)}

The $2K$ selected relays request data transmission to their
desired source nodes. Each source ($\mathcal{S}_k$) then starts to
transmit data to the corresponding destination ($\mathcal{D}_k$)
via one of its two relay nodes alternately
($\mathcal{R}_{\pi_1(k)}$ or $\mathcal{R}_{\pi_2(k)}$), which was
specified earlier. If the TILs of the selected relays are
arbitrarily small, then i) the associated undesired source--relay
and relay--destination channel links and ii) the inter-relay
channel links are all in deep fade. In Section~\ref{subsec:ana},
we will show that it is possible to choose such relays with the
help of the multiuser diversity gain.

At the receiver side, each relay or destination detects the signal
sent from its desired transmitter, while simply treating
interference as Gaussian noise. Thus, no multiuser detection is
performed at each receiver, thereby resulting in an easier
implementation.

%In Fig.~\ref{fig:diagram}, the overall procedure of the OND scheme
%is illustrated again according to the sequence of time slots.
%
%
%\begin{figure}[t]
%    \centerline{\includegraphics[width=0.46\textwidth]{Figs/diagram.eps}}
%    \caption{A block diagram for the OND scheme according to the sequence of time slots.
%    }
%    \label{fig:diagram}
%\end{figure}

%%%%%%%%%%%%%%%%%%%%%%%%%%%%%%%%%%%%%%%%%%%%%%%%%%%%%%%%%%%%%%%%%%%%%%%%%%%%%%%%%%%%%%%%%%%%%%%%%%%%%%%%%%%%%%%%%%%%%%%%%%%%%%%%%%%%%%%%%%%%%%%%%%%%%

\subsection{Analysis of a Lower Bound on the DoF} \label{subsec:ana}

In this subsection, using the scaling argument bridging between
the number of relays, $N$, and the received SNR (refer
to~\cite{JPS:12:COM,H_Yang13_TWC,Codebook_OIA,ODIA_TSP14} for the
details), we shall show 1) the lower bound on the DoF of the
$K\times N \times K$ channel with interfering relays as $N$
increases and 2) the minimum $N$ required to guarantee the
achievability result. The total number of DoF, denoted by
$\textsf{DoF}_\textsf{total}$, is defined
as~\cite{Jafar_IA_original}
\begin{equation}
\textsf{DoF}_\textsf{total} = \sum_{k=1}^K
\left(\lim_{\snr\rightarrow\infty}\frac{T_k(\snr)}{\log
\snr}\right), \nonumber
\end{equation}
where $T_k(\snr)$ denotes the transmission rate of source
$\mathcal{S}_k$. Using the OND framework in the $K\times N\times
K$ channel with interfering relays where $L$ transmission slots
per block are used, the achievable $\textsf{DoF}_\textsf{total}$
is lower-bounded by
\begin{align}
&\textsf{DoF}_\textsf{total}\ge \frac{L-1}{L}\sum_{k=1}^K
\sum_{s=1}^2 \left(\lim_{\snr\rightarrow\infty}
\frac{\frac{1}{2}\log\left(1+\min\left\{\sinr_{\pi_s(k)}^{(1)},\sinr_{k,\pi_s(k)}^{(2)}\right\}\right)}{\log
\snr}\right), \label{EQ:DoF_OND}
\end{align}
where $\sinr_{\pi_s(k)}^{(1)}$ denotes the received
signal-to-interference-and-noise ratio (SINR) at the relay
$\mathcal{R}_{\pi_s(k)}$ and $\sinr_{k,\pi_s(k)}^{(2)}$ denotes
the received SINR at the destination $\mathcal{D}_k$ when the
relay $\mathcal{R}_{\pi_s(k)}$ transmits the desired signal
($s=1,2$ and $k\in\{1,\cdots,K\}$). More specifically, the above
SINRs can be formally expressed as\footnote{Note that at the first
time slot for the relays $\{\pi_1(i)\}_{i=1}^{K}$, the third term
in the denominator of $\sinr_{\pi_1(i)}^{(1)}$ (i.e., the
inter-relay interference term) becomes zero.}
\begin{align*}
   &\sinr_{\pi_s(i)}^{(1)}
    =
    \frac{
        P\left|h_{\pi_s(i)i}^{(1)}\right|^2
    }{
        N_0
        +
        P
        \sum_{k=1 \atop k \neq i}^K
            \Norm{h_{\pi_s(i)k}^{(1)}}^2
        +
        P
        \sum_{k=1}^K
            \Norm{h_{\pi_s(i),\pi_{\tilde{s}}(k)}^{(r)}}^2
    } \\
      &\sinr_{i,\pi_s(i)}^{(2)}
    =
        \frac{
            P\left|h_{i\pi_s(i)}^{(2)}\right|^2
        }{
            1
            +
            P
            \sum_{k=1\atop k\neq i}^K
            \Norm{h_{i\pi_s(k)}^{(2)}}^2
        },
\end{align*}
where the second term in the denominator of
$\sinr_{\pi_s(i)}^{(1)} $ indicates the interference power at
relay $\pi_s(i)$ received from the sources while the third term
indicates the inter-relay interference, and  the second term in
the denominator of $\sinr_{i,\pi_s(i)}^{(2)}$ indicates the
interference power at the destination $\mathcal{D}_i$ received
from the active relays. Here, $\tilde{s}=3-s$, i.e., $\tilde{s}=2$
if $s=1$, and vice versa.

We focus on the first relay set ${\bf \Pi}_1$'s perspective to
examine the received SINR values according to each time slot. Let
us first denote $L_{\pi_1(i),i}^{\mathbf{\Pi_1}} \triangleq
\tilde{L}_{\pi_1(i),i}+
\sum_{k=1}^K\Norm{h_{\pi_1(i),\pi_2(k)}^{(r)}}^2$ for
$i\in\{1,\cdots,K\}$. For the first hop, at time slot $2t-1$
(i.e., each odd time slot), $t \in \left\{1,2,\ldots,
\frac{L-1}{2}\right\}$, the received $\sinr_{\pi_1(i)}^{(1)}$ at
$\mathcal{R}_{\pi_1(i)}^{(1)}$ is lower-bounded by
\begin{align}
   \sinr_{\pi_1(i)}^{(1)} &\ge
    \frac{
        \snr\Norm{h_{\pi_1(i)i}^{(1)}}^2
    }{
        1
        \!\!+\!\!
        \snr
        \sum_{k=1\atop k\neq i}^K \!\!
        \left(
            \Norm{h_{\pi_1(i)k}^{(1)}}^2
           \!\! +\!\!
            \Norm{h_{k\pi_1(i)}^{(2)}}^2
        \right)
       \!\! + \!\!
        \snr
        \sum_{k=1}^K \!\!
            \Norm{h_{\pi_1(i),\pi_2(k)}^{(r)}}^2
    }
    \nonumber\\
    &=
    \frac{
        \snr
        \Norm{h_{\pi_1(i)i}^{(1)}}^2
    }{
        1
        +
        \snr
        \left(
            \tilde{L}_{\pi_1(i),i}
            +
            \sum_{k=1}^K
                \Norm{h_{\pi_1(i),\pi_2(k)}^{(r)}}^2
        \right)
    }
       \nonumber\\
    &=\frac{
        \snr
        \Norm{h_{\pi_1(i)i}^{(1)}}^2
    }{
        1
        +
        \snr
       L_{\pi_1(i),i}^{\mathbf{\Pi_1}}
    }
    \nonumber\\
    &\ge\frac{
        \snr
        \Norm{h_{\pi_1(i)i}^{(1)}}^2
    }{
        1
        +\snr
        \sum_{i=1}^K
       L_{\pi_1(i),i}^{\mathbf{\Pi_1}}
    },
\label{EQ:SINR_1_pi1}
\end{align}
where $\tilde{L}_{\pi_1(i),i}$ indicates the scheduling metric in
(\ref{eq:IL}) when $\mathcal{R}_{\pi_1(i)}$ is assumed to serve
the $i$th S--D pair ($\mathcal{S}_i$, $\mathcal{D}_i$). For the
second hop, at time slot $2t$ (i.e., each even time slot), $t \in
\left\{1,2,\ldots, \frac{L-1}{2}\right\}$, the received
$\sinr_{i,\pi_1(i)}^{(2)}$ at $\mathcal{D}_{i}$ is lower-bounded
by
\begin{align}
    \sinr_{i,\pi_1(i)}^{(2)}
    &\ge
        \frac{
            \snr
            \Norm{h_{i\pi_1(i)}^{(2)}}^2
        }{
            1
            +
            \snr
            \sum_{i=1}^K
            \sum_{k=1\atop k\neq i}^K
                \Norm{h_{i\pi_1(k)}^{(2)}}^2 } \nonumber\\
&\ge
    \frac{
        \snr
        \Norm{h_{i\pi_1(i)}^{(2)}}^2
    }{
        1
        +
        \snr
        \sum_{i=1}^K
        \sum_{k=1 \atop k\neq i}^K
            \left(
                \Norm{h_{\pi_1(i)k}^{(1)}}^2
                +
                \Norm{h_{k\pi_1(i)}^{(2)}}^2
            \right)
    }
 \nonumber\\
    &=
    \frac{
        \snr
        \Norm{h_{i\pi_1(i)}^{(2)}}^2
    }{
        1
        +
        \snr
        \sum_{i=1}^K
            \LIF{i}{i}
    }
    \nonumber\\
     &\ge
    \frac{
        \snr
        \Norm{h_{i\pi_1(i)}^{(2)}}^2
    }{
        1
        +
        \snr
        \sum_{i=1}^K
            L_{\pi_1(i),i}^{\mathbf{\Pi_1}}
    },
\label{EQ:SINR_2_pi1}
\end{align}
where the second inequality holds due to the channel reciprocity.
The term $\sum_{i=1}^K L_{\pi_1(i),i}^{\mathbf{\Pi_1}}$ in the
denominator of (\ref{EQ:SINR_1_pi1}) and (\ref{EQ:SINR_2_pi1})
needs to scale as $\snr^{-1}$, i.e., $\sum_{i=1}^K
L_{\pi_1(i),i}^{\mathbf{\Pi_1}}=O(\snr^{-1})$, so that both
$\sinr_{\pi_1(k)}^{(1)}$ and $\sinr_{k,\pi_1(k)}^{(2)}$ scale as
$\Omega(\snr)$ with increasing SNR, which eventually enables to
achieve the DoF of $\frac{L-1}{L}$ per S--D pair from
(\ref{EQ:DoF_OND}).\footnote{We use the following notation: i)
$f(x)=O(g(x))$ means that there exist constant $C$ and $c$ such
that $f(x)\leq Cg(x)$ for all $x>c$. ii) $f(x)=\Omega(g(x))$ if
$g(x)=O(f(x))$. iii) $f(x)=\omega(g(x))$ means that
$\lim_{x\rightarrow \infty}\frac{g(x)}{f(x)}=0$~\cite{Knuth}.}
Even if such a bounding technique in (\ref{EQ:SINR_1_pi1}) and
(\ref{EQ:SINR_2_pi1}) leads to a loose lower bound on the SINR, it
is sufficient to prove our achievability result in terms of DoF
and relay scaling law.

Now, let us turn to the second relay set ${\bf \Pi}_2$. Similarly
as in (\ref{EQ:SINR_1_pi1}), for the first hop, at time slot $2t$,
$t \in \left\{1,2,\ldots, \frac{L-1}{2}\right\}$, the received
$\sinr_{\pi_2(i)}^{(1)}$ at $\mathcal{R}_{\pi_2(i)}^{(1)}$ is
lower-bounded by
\begin{align}
    \textsf{sinr}_{\pi_2(i)}^{(1)} \ge\frac{
        \snr
        \Norm{h_{\pi_2(i)i}^{(1)}}^2
    }{
        1
        +\snr
        \sum_{i=1}^K
       L_{\pi_2(i),i}^{\mathbf{\Pi_2}}
    }, \label{EQ:SINR_1_pi2}
\end{align}
where $L_{\pi_2(i),i}^{\mathbf{\Pi_2}}$ indicates the TIL in
\eqref{eq:TIL} when $\mathcal{R}_{\pi_2(i)}$ is assumed to serve
the $i$th S--D pair ($\mathcal{S}_i$, $\mathcal{D}_i$). For the
second hop, at time slot $2t+1$, $t \in \left\{1,2,\ldots,
\frac{L-1}{2}\right\}$, the received $\sinr_{i,\pi_2(i)}^{(2)}$ at
$\mathcal{D}_{i}$ can also be lower-bounded by
\begin{align}
\sinr_{i,\pi_2(i)}^{(2)} \ge
\frac{\snr\left|h_{i\pi_2(i)}^{(2)}\right|^2}{1+\snr\sum_{i=1}^K
L_{\pi_2(i),i}^{\mathbf{\Pi_2}}}. \label{EQ:SINR_2_pi2}
\end{align}
%Here, $\tilde{L}_{\pi_2(i),i}$ denotes the scheduling metric in
%\eqref{eq:IL} when $\mathcal{R}_{\pi_2(i)}$ is assumed to serve
%the $i$th S--D pair.
%Here, the term $\sum_{i=1}^KL_{\pi_2(i),i}$ in the denominator of
%\eqref{EQ:SINR_2_pi2} needs to scale as $\snr^{-1}$, i.e.,
%$\sum_{i=1}^K L_{\pi_2(i),i}=O(\snr^{-1})$, so that both
%$\sinr_{\pi_2(k)}^{(1)}$ and $\sinr_{k,\pi_2(k)}^{(2)}$
%scale as $O\left(\snr\right)$ with increasing SNR.

The next step is thus to characterize the three metrics
$\tilde{L}_{i,k}$, $L_{i,k}^{\mathbf{\Pi_1}}$, and
$L_{i,k}^{\mathbf{\Pi_2}}$ ($i\in\{1,\ldots,N\}$ and
$k\in\{1,\ldots,K\}$) and their cumulative density functions
(CDFs) in the $K\times N\times K$ channel with interfering relays,
which is used to analyze the lower bound on the DoF and the
required relay scaling law in the model under consideration. Since
it is obvious to show that the CDF of $L_{i,k}^{\mathbf{\Pi_1}}$
is identical to that of $L_{i,k}^{\mathbf{\Pi_2}}$, we focus only
on the characterization of $L_{i,k}^{\mathbf{\Pi_2}}$. The
scheduling metric $\tilde{L}_{i,k}$ follows the chi-square
distribution with $2(2K-2)$ degrees of freedom since it represents
the sum of i.i.d. $2K-2$ chi-square random variables with 2
degrees of freedom. Similarly, the TIL $L_{i,k}^{\mathbf{\Pi_2}}$
follows the chi-square distribution with $2(3K-2)$ degrees of
freedom. The CDFs of the two metrics $\tilde{L}_{i,k}$ and
$L_{i,k}^{\mathbf{\Pi_2}}$ are given by
\begin{align}
    \mathcal{F}_{\tilde{L}}\left(\ell\right)=\frac{\gamma(2K-2,\ell/2)}{\Gamma(2K-2)}
    \label{eq:Ltilde:cdf}
\\
    \mathcal{F}_{L}\left(\ell\right)=\frac{\gamma(3K-2,\ell/2)}{\Gamma(3K-2)},
\label{eq:L:cdf}
\end{align}
respectively, where $\Gamma(z)=\int_0^{\infty}t^{z-1}e^{-t}dt$ is
the Gamma function and $\gamma(z,x)=\int_0^x t^{z-1}e^{-t}dt$ is
the lower incomplete Gamma
function~\cite[eqn.~(8.310.1)]{GradshteynRyzhik}. We start from
the following lemma.

\begin{lemma} \label{lem:1}
For any $0< \ell \leq 2$, the CDFs of the random variables
$\tilde{L}_{i,k}$ and $L_{i,k}^{\mathbf{\Pi_2}}$  in
\eqref{eq:Ltilde:cdf} and \eqref{eq:L:cdf} are lower-bounded by
$\mathcal{F}_{\tilde{L}}\left(\ell\right)\ge C_1 \ell^{2K-2}$ and
$\mathcal{F}_{L}\left(\ell\right)\ge C_2 \ell^{3K-2}$,
respectively, where
\begin{align}\label{eq:C1}
    C_1= \frac{e^{-1}2^{-(2K-2)}}{\Gamma(2K-1)}
    \\
    C_2= \frac{e^{-1}2^{-(3K-2)}}{\Gamma(3K-1)}, \label{eq:C2}
\end{align}
and $\Gamma(z)$ is the Gamma function.
\end{lemma}

\begin{proof}
The detailed proof of this argument is omitted here since it
essentially follows the similar line to the proof of~\cite[Lemma
1]{JPS:12:COM} with a slight modification.
\end{proof}

In the following theorem, we establish our first main result by
deriving the lower bound on the total DoF in the $K\times N \times
K$ channel with interfering relays.
\begin{theorem} \label{thrm:1}
Suppose that the OND scheme with alternate relaying is used for
the $K\times N \times K$ channel with interfering relays. Then,
for $L$ data transmission time slots,
\begin{equation}
\textsf{DoF}_\textsf{total} \ge \frac{(L-1)K}{L} \nonumber
\end{equation}
is achievable if $N = \omega\left( \snr^{3K-2}\right)$.
\end{theorem}

\begin{proof}
From (\ref{EQ:DoF_OND})--(\ref{EQ:SINR_2_pi2}), the OND scheme
achieves $\textsf{DoF}_\textsf{total}\ge \frac{L-1}{L}K$ provided
that the two values $\snr\sum_{i=1}^K
L_{\pi_1(i),i}^{\mathbf{\Pi_1}}$ and $\snr\sum_{i=1}^K
L_{\pi_2(i),i}^{\mathbf{\Pi_2}}$ are less than or equal to some
constant $\epsilon_0>0$, independent of SNR, for all S--D pairs.
Then, a lower bound on the achievable
$\textsf{DoF}_\textsf{total}$ is given by
\begin{align}
\textsf{DoF}_\textsf{total} \ge
\mathcal{P}_{\mathrm{OND}}\frac{(L-1)K}{L}, \nonumber
\end{align}
which indicates that $\frac{L-1}{L}K$ DoF is achievable for a
fraction $\mathcal{P}_{\mathrm{OND}}$ of the time for actual
transmission, where
\begin{align}\label{eq:POND:def0}
    &\mathcal{P}_{\mathrm{OND}} =
        \lim_{\snr \to \infty} \!\!
                \Pro{
                    \snr
                    \sum_{i=1}^K L_{\pi_1(i),i}^{\mathbf{\Pi_1}}
                    \leq
                    \epsilon_0
                \textrm{ and }
                    \snr
                    \sum_{i=1}^K  L_{\pi_2(i),i}^{\mathbf{\Pi_2}}
                    \leq
                    \epsilon_0
                }.
\end{align}

We now examine the relay scaling condition such that
$\mathcal{P}_{\mathrm{OND}}$ converges to one with high
probability. For the simplicity of the proof, suppose that the
first and the second relay sets ${\bf \Pi}_1$ and ${\bf \Pi}_2$
are selected out of two mutually exclusive relaying candidate sets
$\mathcal{N}_1$ and $\mathcal{N}_2$, respectively, i.e.,
$\mathcal{N}_1, \mathcal{N}_2\subset \{1,\ldots, N\}$,
$\mathcal{N}_1\cap \mathcal{N}_2 = \emptyset$, $\mathcal{N}_1 \cup
\mathcal{N}_2 = \{1,\ldots, N\}$, ${\bf
\Pi}_1\subset\mathcal{N}_1$, and ${\bf
\Pi}_2\subset\mathcal{N}_2$. Then, we are interested in how
$|\mathcal{N}_1|$ and $|\mathcal{N}_2|$ scale with SNR in order to
guarantee that $\mathcal{P}_{\mathrm{OND}}$ tends to one, where
$|\mathcal{N}_s|$ denotes the cardinality of $\mathcal{N}_s$ for
$s=1,2$. From (\ref{eq:POND:def0}), we further have
\begin{align}\label{eq:POND:def}
    \mathcal{P}_{\mathrm{OND}}
    =
        \lim_{\snr \to \infty}
           & \Bigg(
                \Pro{
                    \snr
                    \sum_{i=1}^K L_{\pi_1(i),i}^{\mathbf{\Pi_1}}
                    \leq
                    \epsilon_0
 %               \Bigg| {\bf \Pi}_1
                     }
                \Pro{
                    \snr
                    \sum_{i=1}^K  L_{\pi_2(i),i}^{\mathbf{\Pi_2}}
                    \leq
                    \epsilon_0
                }
            \Bigg).
\end{align}

Let $\CanSB_m\triangleq \mathcal{N}_2 \setminus
\{\pi_{2}(\ell)\}_{\ell=1}^{m-1}$ with
$\{\pi_{2}(\ell)\}_{\ell=1}^{0}=\emptyset$ and $\Norm{\CanSB_m}$
be the candidate set associated with the second relay set and the
$m$th S-D pair and its cardinality, respectively. For a constant
$\epsilon_0>0$, we can bound the second term in
\eqref{eq:POND:def} as follows:
\begin{align}\label{eq:POND:2}
   & \Pro{
        \sum_{i=1}^K
            L_{\pi_2(i),i}^{\mathbf{\Pi_2}}
    \leq
        \frac{\epsilon_0}{\snr}
        }     \nonumber\\
    & \geq
        \Pro{
            \max_{1\leq i \leq K} L_{\pi_2(i),i}^{\mathbf{\Pi_2}}
            \leq
            \frac{\epsilon_0}{K\snr}
        }
    \nonumber\\
    &\mathop \geq \limits^{\mathrm{(a)}}
        1
        -
        \Pro{
            \exists i: L_{\pi_2(i),i}^{\mathbf{\Pi_2}}
            \geq
            \frac{\epsilon_0}{K\snr}
        }
    \nonumber\\
    &\mathop \geq \limits^{\mathrm{(b)}}
        1
        -
        \sum_{i=1}^{K}
            \Pro{
                L_{\pi_2(i),i}^{\mathbf{\Pi_2}}
                \geq
                \frac{\epsilon_0}{K\snr}
            }
    \nonumber\\
    &\mathop \geq \limits^{\mathrm{(c)}}
        1
        -
        K
        \Pro{
            \min_{j\in\mathcal{N}_2}L_{j,i}^{\mathbf{\Pi_2}}
                \geq
                \frac{\epsilon_0}{K\snr}
        }
    \nonumber\\
    &\mathop \geq \limits^{\mathrm{(d)}}
        1
        -
        K
        \left(
            1
            -
            \mathcal{F}_{L}
            \left(
                \frac{\epsilon_0}{K\snr}
            \right)
        \right)^{\Norm{\CanSB_i}}
    \nonumber\\
    &\mathop \geq \limits^{\mathrm{(e)}}
        1
        -
        K
        \left(
            1
            -
            C_2
            \left(
                \frac{\epsilon_0}{K\snr}
            \right)^{3K-2}
        \right)^{\Norm{\mathcal{N}_2}-K+1},
\end{align}
where the inequality $\mathrm{(a)}$ holds from the De Morgan's
law; $\mathrm{(b)}$ follows from the union bound; $\mathrm{(c)}$
follows since
$L_{\pi_2(i),i}^{\mathbf{\Pi_2}}=\min_{j\in\mathcal{N}_2}L_{j,i}^{\mathbf{\Pi_2}}$;
$\mathrm{(d)}$ follows since $L_{j,i}^{\mathbf{\Pi_2}}$ are the
i.i.d. random variables $\forall j\in \mathcal{N}_2$ for a given
$i$, owning to the fact that the channels are i.i.d. variables;
and $\mathrm{(e)}$ follows from Lemma~\ref{lem:1} with
$C_2=\frac{e^{-1}2^{-(3K-2)}}{\Gamma(3K-1)}$ since
$0<\frac{\epsilon_0}{\snr}\leq 2$  as $\snr \to \infty$ and from
the fact that $\Norm{\CanSB_i}\geq \Norm{\mathcal{N}_2}-K+1$.

We now pay our attention to the first term in \eqref{eq:POND:def},
which can be bounded by
\begin{align}\label{eq:POND:1}
    \Pro{
        \sum_{i=1}^K
            L_{\pi_1(i),i}^{\mathbf{\Pi_1}}
        \leq
        \frac{\epsilon_0}{\snr}
%        \Bigg| {\bf \Pi}_1
        }
    &\geq
        \Pro{\max_{1\leq i \leq K}L_{\pi_1(i),i}^{\mathbf{\Pi_1}} \leq \frac{\epsilon_0}{K\snr}}
    \nonumber\\
     &=
        \left(\Pro{L_{\pi_1(i),i}^{\mathbf{\Pi_1}} \leq \frac{\epsilon_0}{K\snr}}
        \right)^{K},
\end{align}
%\begin{align}
%    \nonumber\\
%   &=
%       \left(\Pro{\min_{j\in\mathcal{N}_1}\tilde{L}_{j,i} \leq \frac{\epsilon_0}{K\snr}} \right)^{K}       \nonumber\\
%   &=
%       \left(
%           1
%           -
%           \left(
%               1
%               -
%               \mathcal{F}_{\tilde{L}}
%               \left(
%                   \frac{\epsilon_0}{K\snr}
%               \right)
%           \right)^{\Norm{\mathcal{N}_1}}
%       \right)^K
%   \nonumber\\
%%  &\mathop \geq \limits^{\mathrm{(a)}}
%%      1
%%      -
%%      \Norm{\mathcal{N}_1}^{K-1}
%%      \sum_{i=0}^{K-1}
%%          \left(
%%              1
%%              -
%%              \mathcal{F}_{\tilde{L}}
%%              \left(
%%                  \frac{\epsilon_0}{K\snr}
%%              \right)
%%          \right)^{\Norm{\mathcal{N}_1}-i}
%%  \nonumber\\
%   &\mathop \geq \limits^{\mathrm{(b)}}
%   \left(
%       1
%       -
%           \left(
%               1
%               -
%               C_1
%               \left(
%                   \frac{\epsilon_0}{K\snr}
%               \right)^{2K-2}
%           \right)^{\Norm{\mathcal{N}_1}}
%   \right)^{K},
%\end{align}
where the equality follows from the fact that
$L_{\pi_1(i),i}^{\mathbf{\Pi_1}}$ and
$L_{\pi_1(j),j}^{\mathbf{\Pi_1}}$ for $i \neq j$ are the functions
of different random variables and thus are independent of each
other.
%$\mathrm{(b)}$ follows from Lemma~\ref{lem:1} with  $C_1=\frac{e^{-1}2^{-(2K-2)}}{\Gamma(2K-1)}$.
By letting
$K_{i}=\sum_{k=1}^K\Norm{h_{\pi_1(i),\pi_2(k)}^{(r)}}^2$, by the
definition of $L_{\pi_1(i),i}^{\mathbf{\Pi_1}}$, we have
\begin{align}\label{eq:TIF}
    \Pro{L_{\pi_1(i),i}^{\mathbf{\Pi_1}} \leq
    \frac{\epsilon_0}{K\snr}} &=
        1
        -
        \Pro{\LIF{i}{i}+K_{i}\geq \frac{\epsilon_0}{K\snr}}
    \nonumber\\
    &\geq
        1
        -
        \Pro{
            \LIF{i}{i}\geq \frac{\epsilon_0}{2K\snr}
        } -
        \Pro{
            K_{i}\geq \frac{\epsilon_0}{2K\snr}
        },
\end{align}
where the inequality follows from the fact that for any random
variables $X$ and $Y$, $\Pro{X+Y \geq \epsilon}$ $\leq \Pro{X \geq
\frac{\epsilon}{2}}+ \Pro{Y \geq \frac{\epsilon}{2}}$
\cite{LB:11:book}. In the same manner, let $\CanSA_m\triangleq
\mathcal{N}_1 \setminus \{\pi_{1}(\ell)\}_{\ell=1}^{m-1}$ with
$\{\pi_{1}(\ell)\}_{\ell=1}^{m-1}=\emptyset$ and $\Norm{\CanSA_m}$
be the candidate set associated with the first relay set and the
$m$th S--D pair and its cardinality, respectively. Then, we can
bound the first two terms in the RHS of~\eqref{eq:TIF} as follows:
\begin{align}\label{eq:POND:1}
    1
        -
        \Pro{
            \LIF{i}{i}\geq \frac{\epsilon_0}{2K\snr}
        }
%   &\geq
%       \Pro{\max_{1\leq i \leq K}\LIF{i}{i} \leq \frac{\epsilon_0}{2K\snr}}
%    \nonumber\\
%    &\mathop \geq \limits^{\mathrm{(a)}}
%       \left(\Pro{\LIF{i}{i} \leq \frac{\epsilon_0}{2K\snr}} \right)^{K}
%    \nonumber\\
    &=
        1
        -
        \Pro{\min_{j\in\mathcal{N}_1}\tilde{L}_{j,i} \geq \frac{\epsilon_0}{2K\snr}}
    \nonumber\\
    &=
        1
        -
        \left(
            1
            -
            \mathcal{F}_{\tilde{L}}
            \left(
                \frac{\epsilon_0}{2K\snr}
            \right)
        \right)^{\Norm{\CanSA_i}}
    \nonumber\\
    &\geq
        1
        -
            \left(
                1
                -
                C_1
                \left(
                    \frac{\epsilon_0}{2K\snr}
                \right)^{2K-2}
            \right)^{\Norm{\mathcal{N}_1}-K+1},
\end{align}
%where the equality $\mathrm{(a)}$ follows from the fact that for $\LIF{i}{i}$ and $\LIF{j}{j}$ with $i \neq j$ are the functions of different random vectors, and thus are independent of each other; $\mathrm{(b)}$ follows from Lemma~\ref{lem:1} with  $C_1=\frac{e^{-1}2^{-(2K-2)}}{\Gamma(2K-1)}$.
where the last inequality follows from Lemma~\ref{lem:1} with
$C_1=\frac{e^{-1}2^{-(2K-2)}}{\Gamma(2K-1)}$. Finally, from
\eqref{eq:POND:2}, it follows that $\Pro{K_{i}\geq
\frac{\epsilon_0}{2K\snr}}$ tends to zero as
$\Norm{\mathcal{N}_2}$ grows large by noting that
$L_{\pi_2(i),i}^{\mathbf{\Pi_2}}=\tilde{L}_{\pi_2(i),i}+K_i$ due
to the reciprocal property of TDD systems. From \eqref{eq:POND:2},
\eqref{eq:TIF}, and \eqref{eq:POND:1}, it is obvious that if
$\Norm{\mathcal{N}_1}$ and $\Norm{\mathcal{N}_2}$ scale faster
than $\snr^{2K-2}$ and $\snr^{3K-2}$, respectively, then
\begin{align}
    \lim_{\snr\to \infty}
        \left(
            1
            -
            C_1
            \left(
                \frac{\epsilon_0}{2K\snr}
            \right)^{2K-2}
        \right)^{\Norm{\mathcal{N}_1}-K+1}
    = 0
    \\
    \lim_{\snr\to \infty}
        \left(
            1
            -
            C_2
            \left(
                \frac{\epsilon_0}{K\snr}
            \right)^{3K-2}
        \right)^{\Norm{\mathcal{N}_2}-K+1}
    = 0.
\end{align}
%\begin{itemize}
%\item $|\mathcal{N}_1|=\omega\left( \snr^{2K-2}\right)$
%\item $|\mathcal{N}_2|=\omega\left( \snr^{3K-2}\right)$
%\item $N = |\mathcal{N}_1| + |\mathcal{N}_2| = \omega\left(
%\snr^{3K-2}\right)$
%\end{itemize}

%\begin{align}
%\mathcal{P}_{\mathrm{FD-OND}} &\ge \lim_{\mathrm{SNR} \to \infty}
%\mathrm{Pr} \left\{ L_{\pi_s(i),i} \le \frac{\epsilon_0
%\snr^{-1}}{K}, ~\forall k \textrm{~and~} \forall s\right\}
%\nonumber\\ & = \lim_{\mathrm{SNR} \to \infty} \prod_{k=1}^{K}
%\prod_{s=1}^2 \mathrm{Pr} \left\{ L_{\pi_s(i),i} \le
%\frac{\epsilon_0 \snr^{-1}}{K}\right\} \nonumber\\ & \ge
%\lim_{\mathrm{SNR} \to \infty} \mathrm{Pr} \left(\left\{ \min_{k,
%s}L_{\pi_s(i),i} \le \frac{\epsilon_0
%\snr^{-1}}{K}\right\}\right)^{2K} \nonumber\\
%&\ge \lim_{\mathrm{SNR} \to \infty} \mathrm{Pr} \left(\left\{
%\min_{k, s}L_{\pi_s(i),i} \le \frac{\epsilon_0
%\snr^{-1}}{K}\right\}\right)^{2K}
%\end{align}
Therefore, $\mathcal{P}_{\mathrm{OND}}$ asymptotically approaches
one, which means that the DoF of $\frac{(L-1)K}{L}$ is achievable
with high probability if $N = \Norm{\mathcal{N}_1}+
\Norm{\mathcal{N}_2}=\omega\left( \snr^{3K-2}\right)$. This
completes the proof of the theorem.
\end{proof}

Note that the lower bound on the DoF asymptotically approaches $K$
for large $L$, which implies that our system operates in virtual
full-duplex mode. The parameter $N$ required to obtain full DoF
(i.e., $K$ DoF) needs to increase exponentially with the number of
S--D pairs, $K$, in order to make the sum of $3K-2$ interference
terms in the TIL metric (\ref{eq:TIL}) non-increase with
increasing SNR at each relay.\footnote{Since $\snr^{2K-2}$ scales
slower than $\snr^{3K-2}$, it does not affect the performance in
terms of DoF and relay scaling laws.} Here, from the perspective
of each relay in $\mathbf{\Pi}_2$, the SNR exponent $3K-2$
indicates the total number of interference links and stems from
the following three factors: the sum of interference power
received from other sources, the sum of interference power
generated to other destinations, and the sum of inter-relay
interference power generated from the relays in $\mathbf{\Pi}_1$.
From Theorem~\ref{thrm:1}, let us provide the following
interesting discussions regarding the DoF achievability.

\begin{remark}
$K$ DoF can be achieved by using the proposed OND scheme in the
$K\times N\times K$ channel with interfering links among relay
nodes, if the number of relay nodes, $N$, scales faster than
$\snr^{3K-2}$ and the number of transmission slots in one block,
$L$, is sufficiently large. In this case, all the interference
signals are almost nulled out at each selected relay by exploiting
the multiuser diversity gain. In other words, by applying the OND
scheme to the interference-limited $K\times N\times K$ channel
such that the channel links are inherently coupled with each
other, the links among each S--D path via one relay can be
completely decoupled, thus enabling us to achieve the same DoF as
in the interference-free channel case.
\end{remark}

%\begin{remark}
%It is worth noting that there exists a trade-off between the
%achievable DoF and the minimum required $N$ by comparing the two
%proposed schemes with and without alternate relaying. It is seen
%that the achievable DoF get decreased to $K/2$ with no alternate
%relaying while $N$ is also significantly reduced to $SNR^{2(K-1)}$.
%This comes from the fact that there is no inter-relay interference
%for the ORS-IN with no alternate relaying, where the first relay set
%$\mathbf{\Pi_1}$ only participates in data forwarding.
%\end{remark}

\begin{remark}
It is not difficult to show that the centralized relay selection
method that maximizes the received SINR (at either the relay or
the destination) using global CSI at the transmitter, which is a
combinatorial problem with exponential complexity, gives the same
relay scaling result $N = \omega\left( \snr^{3K-2}\right)$ along
with full DoF. However, even with our OND scheme using a
decentralized relay selection based only on local CSI, the same
achievability result is obtained, thus resulting in a much easier
implementation.
\end{remark}

%%%%%%%%%%%%%%%%%%%%%%%%%%%%%%%%%%%%%%%%%%%%%%%%%%%%%%%%%%%%%%%%%%%%%%%%%%%%%%%%%%%%%%%%%%%%%%%%%%%%%%%%%%%%%%%%%%%%%%%%%%%%%%%%%%%%%%%%%%%%%%%%%%%%%

\subsection{The TIL Decaying Rate}

In this subsection, we analyze the TIL decaying rate under the OND
scheme with alternate relaying, which is meaningful since the
desired relay scaling law is closely related to the TIL decaying
rate with respect to $N$ for given SNR.

Let $L_{K\text{th-min}}$ denote the $K$th smallest TIL among the
ones that $N$ selected relay nodes compute. Since the $K$ relays
yielding the TIL values up to the $K$th smallest one are selected,
the $K$th smallest TIL is the largest among the TILs that the
selected relays compute. Similarly as in~\cite{Jose12}, by
Markov's inequality, a lower bound on the average decaying rate of
$L_{K\text{th-min}}$ with respect to $N$,
$\mathbb{E}\left[\frac{1}{L_{K\text{th-min}}}\right]$, is given by
\begin{align}
\mathbb{E}\left[\frac{1}{L_{K\text{th-min}}}\right]\ge
\frac{1}{\epsilon}\Pr(L_{K\text{th-min}}\le \epsilon),
\end{align}
where the inequality always holds for $\epsilon>0$. We denote
$\mathcal{P}_K(\epsilon)$ as the probability that there are only
$K$ relays satisfying $\text{TIL}\le \epsilon$, which is expressed
as
\begin{align}
\mathcal{P}_K(\epsilon)= {{N}\choose{K}} \mathcal{F}_L(\epsilon)^K
(1-\mathcal{F}_L(\epsilon))^{N-K}, \label{EQ:P_K}
\end{align}
where $\mathcal{F}_L(\epsilon)$ is the CDF of the TIL. Since
$\Pr(L_{K\text{th-min}}\le \epsilon)$ is lower-bounded by
$\Pr(L_{K\text{th-min}}\le \epsilon)\ge \mathcal{P}_K(\epsilon)$,
a lower bound on the average TIL decaying rate is given by
\begin{align}
\mathbb{E}\left[\frac{1}{L_{K\text{th-min}}}\right]\ge
\frac{1}{\epsilon}\mathcal{P}_K(\epsilon). \label{EQ:TIL_min}
\end{align}

The next step is to find the parameter $\hat{\epsilon}$ that
maximizes $\mathcal{P}_K(\epsilon)$ in terms of $\epsilon$ in
order to provide the tightest lower bound.

\begin{lemma} \label{lem:2}
When a constant $\hat{\epsilon}$ satisfies the condition
$\mathcal{F}_L(\hat{\epsilon})=K/N$,
$\mathcal{P}_K(\hat{\epsilon})$ in (\ref{EQ:P_K}) is maximized for
a given $N$.
\end{lemma}

\begin{proof}
To find the parameter $\hat{\epsilon}$ that maximizes
$\mathcal{P}_K(\epsilon)$, we take the first derivative with
respect to $\epsilon$, resulting in
\begin{align}
&\frac{\partial\mathcal{P}_K(\epsilon) }{\partial\epsilon} =
\frac{\partial\mathcal{F}_T(\epsilon) }{\partial\epsilon}
{{N}\choose{K}}\mathcal{F}_T(\epsilon)^{K-1}\left(1-\mathcal{F}_T(\epsilon)\right)^{N-K-1}\!\!\left(K-N\mathcal{F}_T(\epsilon)\right),
\nonumber
\end{align}
which is zero when
\begin{align}
\hat{\epsilon}=\mathcal{F}_T^{-1}\left(\frac{K}{N}\right).
\end{align}
The parameter $\hat{\epsilon}$ is the unique value that maximizes
$\mathcal{P}_K(\epsilon)$ since
\begin{align}
&\frac{\partial\mathcal{P}_K(\epsilon)
}{\partial\epsilon}\bigg|_\epsilon
>0~~\text{if}~0<\epsilon<\hat{\epsilon} \nonumber\\
& \frac{\partial\mathcal{P}_K(\epsilon)
}{\partial\epsilon}\bigg|_\epsilon
>0~~\text{if}~\epsilon\ge \hat{\epsilon}, \nonumber
\end{align}
which completes the proof of the lemma.
\end{proof}

Now, we establish our second main theorem, which shows a lower
bound on the TIL decaying rate with respect to $N$.

\begin{theorem} \label{thrm:2}
Suppose that the OND scheme with alternate relaying is used for
the $K\times N\times K$ channel with interfering relays. Then, the
decaying rate of TIL is lower-bounded by
\begin{align}
\mathbb{E}\left[\frac{1}{L_{K\text{th-min}}}\right] \ge
\Theta\left(N^{\frac{1}{3K-2}}\right).
\end{align}
\end{theorem}

\begin{proof}
As shown in (\ref{EQ:TIL_min}), the TIL decaying rate is
lower-bounded by the maximum of
$\frac{1}{\epsilon}\mathcal{P}_K(\epsilon)$ over $\epsilon$. By
Lemma~\ref{lem:2}, $\mathcal{P}(\hat{\epsilon})$ is maximized when
$\hat{\epsilon}=\mathcal{F}_L^{-1}\left(\frac{K}{N}\right)$. Thus,
we have
\begin{align}
\mathbb{E}\left[\frac{1}{L_{K\text{th-min}}}\right] &\ge
\frac{1}{\mathcal{F}_L^{-1}(K/N)}{{N}\choose{K}}\left(\frac{K}{N}\right)^K
\left(1-\frac{K}{N}\right)^{N-K} \nonumber\\
& \ge
\frac{1}{\mathcal{F}_L^{-1}(K/N)}\left(\frac{N-K+1}{N}\right)^K
\!\!\!\left(1-\frac{K}{N}\right)^{N-K} \nonumber\\
& \ge \frac{1}{\mathcal{F}_L^{-1}(K/N)} \left(\frac{1}{K}\right)^K
e^{-K} \nonumber\\ &\ge \Theta\left(N^{\frac{1}{3K-2}}\right),
\nonumber
\end{align}
where the second and third inequalities hold since
${{N}\choose{K}}\ge \left(\frac{N-K+1}{K}\right)^K$ and
$\left(1-\frac{K}{N}\right)^{N-K} \ge \left(1-\frac{K}{N}\right)^N
\ge e^{-K}$, respectively. By Lemma~\ref{lem:1}, it follows that
$\frac{1}{\mathcal{F}_L^{-1}(K/N)}\ge
\left(\frac{C_1N}{K}\right)^{\frac{1}{3K-2}}$, where $C_1$ is
given by (\ref{eq:C1}). Hence, the last inequality also holds,
which completes the proof of the theorem.
\end{proof}

From Theorem~\ref{thrm:2}, the following valuable insight is
provided: the smaller SNR exponent of the relay scaling law, the
faster TIL decaying rate with respect to $N$. This asymptotic
result will be verified in a finite $N$ regime via numerical
evaluation in Section~\ref{sec:5}.

\subsection{OND Without Alternate Relaying}

For comparison, the OND scheme without alternate relaying is also
explained in this subsection.

It is worth noting that there exists a trade-off between the lower
bound on the DoF and the minimum number of relays required to
guarantee our achievability result by additionally introducing the
OND protocol without alternate relaying. In the scheme, the first
relay set $\mathbf{\Pi}_1$ only participates in data forwarding.
That is, the second relay set $\mathbf{\Pi}_2$ does not need to be
selected for the OND protocol without alternate relaying.
Specifically, the steps of each node during one block are then
described as follows:
\begin{itemize}
\item Time slot 1: Sources $\mathcal{S}_1,\cdots,\mathcal{S}_K$
transmit their first encoded symbols
$x_1^{(1)}(1),\cdots,x_K^{(1)}(1)$, where $x_k^{(1)}(l)$
represents the $l$th transmitted symbol of the $k$th source node.
A set of $K$ selected relay nodes,
$\mathbf{\Pi}_1=\{\pi_1(1),\cdots,\pi_1(K)\}$, operating in
receive mode at each odd time slot, listens to
$x_1^{(1)}(1),\cdots,x_K^{(1)}(1)$. Other $N-K$ relay nodes and
destinations $\mathcal{D}_1,\cdots,\mathcal{D}_K$ remain idle.

\item Time slot 2: The $K$ relays in the set $\mathbf{\Pi}_1$
forward their first re-encoded symbols
$x_{\pi_1(1)}^{(2)}(1),\cdots,$ $x_{\pi_1(K)}^{(2)}(1)$ to the
corresponding $K$ destinations. The $K$ destinations receive from
$\mathcal{R}_{\pi_1(1)},$ $\cdots,\mathcal{R}_{\pi_1(K)}$ and
decode $x_{\pi_1(1)}^{(2)}(1),\cdots,x_{\pi_1(K)}^{(2)}(1)$. The
remaining $N-K$ relays keep idle.

\item The processes in time slots 1 and 2 are repeated to the
$(L-1)$th time slot.

\item Time slot $L$: The $K$ relays in $\mathbf{\Pi}_1$ forward
their re-encoded symbols $x_{\pi_1(1)}^{(2)}(L-1),\cdots,$
$x_{\pi_1(K)}^{(2)}(L-1)$ to the corresponding $K$ destinations.
The $K$ sources and the other $N-K$ relays remain idle.
\end{itemize}

When $\mathcal{R}_i$ is assumed to serve the $k$th S--D pair
$(\mathcal{S}_k, \mathcal{D}_k)$ ($i\in\{1,\cdots,N\}$ and
$k\in\{1,\cdots,K\}$), it computes the scheduling metric
$\tilde{L}_{i,k}$ in (\ref{eq:IL}). According to the computed
$\tilde{L}_{i,k}$, a timer based method is used for relay
selection as in Section~\ref{SubSub:1strelay}. Because there is no
inter-relay interference for the OND scheme without alternate
relaying, it is expected that the minimum required $N$ to achieve
the optimal DoF is significantly reduced. Our third main theorem
is established as follows.

\begin{theorem} \label{thrm:3}
Suppose that the OND scheme without alternate relaying is used for
the $K\times N \times K$ channel. Then, for $L$ data transmission
time slots,
\begin{equation}
\textsf{DoF}_\textsf{total} \ge \frac{K}{2} \nonumber
\end{equation}
is achievable if $N = \omega\left( \snr^{2K-2}\right)$.
\end{theorem}

\begin{proof}
The detailed proof of this argument is omitted here since it
basically follows the same line as the proof of
Theorem~\ref{thrm:1}.
\end{proof}

In Section~\ref{sec:5}, it will be also seen that in a finite $N$
regime, there exists the case even where the OND without alternate
relaying outperforms that of the OND with alternate relaying in
terms of achievable sum-rates via computer simulations.

%%%%%%%%%%%%%%%%%%%%%%%%%%%%%%%%%%%%%%%%%%%%%%%%%%%%%%%%%%%%%%%%%%%%%%%%%%%%%%%%%%%%%%%%%%%%%%%%%%%%%%%%%%%%%%%%%%%%%%%%%%%%%%%%%%%%%%%%%%%%%%%%%%%%%

\section{Upper Bound for DoF} \label{sec:4}

In this section, to show the optimality of the proposed OND scheme
in the $K\times N\times K$ channel with interfering relays, which
consists of $K$ S--D pairs and $N$ relay nodes, we derive an upper
bound on the DoF using the cut-set bound~\cite{CoverThomas:91} as
a counterpart of the lower bound on the total DoF in
Section~\ref{subsec:ana}. Suppose that $\tilde{N}$ relay nodes are
active, i.e., receive packets and retransmit their re-encoded
ones, simultaneously, where $\tilde{N}\in\{1,\cdots,N\}$. This is
a generalized version of our transmission since it is not
characterized how many relays need to be activated simultaneously
to obtain the optimal DoF. We consider the two cuts $L_1$ and
$L_2$ dividing our network into two parts in a different manner.
Let $\mathcal{S}_{L_i}$ and $\mathcal{D}_{L_i}$ denote the sets of
sources and destinations, respectively, for the cut $L_i$ in the
network ($i=1,2$). For the $K\times N\times K$ channel model with
interfering relays, we now use the fact that there is no direct
path between an S--D pair. Then, it follows that under $L_1$, $K$
transmit nodes in $\mathcal{S}_{L_1}$ are on the left of the
network, while $\tilde{N}$ active relay nodes and $K$ (final)
destination nodes in $\mathcal{D}_{L_2}$ are on the right and act
as receivers. In this case, we can create the $K\times
(\tilde{N}+K)$ multiple-input multiple-output (MIMO) channel
between the two sets of nodes separated by the cut. Similarly, the
$(\tilde{N}+K) \times K$ MIMO channel are obtained under the cut
$L_2$. It is obvious to show that DoF for the two MIMO channels is
upper-bounded by $K$. Hence, it turns out that even with the
half-duplex assumption, our lower bound on the DoF based on the
OND with alternate relaying asymptotically approaches this upper
bound on the DoF for large $L$.

%\begin{figure}[t!]
%\centering \subfigure[The cut $L_1$]{
%{\includegraphics[width=0.18\textwidth]{Figs/cut1.eps}}
%\label{fig:cut:1} }~~~~~~ \subfigure[The cut $L_2$]{
%{\includegraphics[width=0.24\textwidth]{Figs/cut2.eps}}
%\label{fig:cut:2} }~ \caption{The two cuts in the $K\times N\times
%K$ channel model with interfering relays.} \label{fig:cut}
%\end{figure}

Note that this upper bound is generally derived regardless of
whether the number of relays, $N$, tends to infinity or not,
whereas the scaling condition $N = \omega\left(
\snr^{3K-2}\right)$ is included in the achievability proof.

\begin{figure}[t!]
    \centerline{\includegraphics[width=0.57\textwidth]{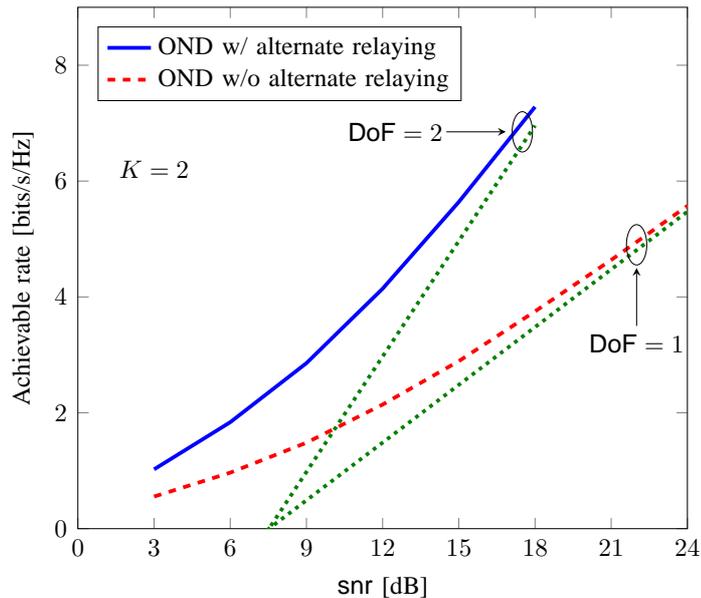}}
    \caption{The achievable sum-rates as a function of $\snr$ when $K=2$ and $N=\snr^{4}$ in the $K\times N \times K$ channel. It is assumed that $N=\snr^{4}$ and $N=\snr^{2}$ are used for the OND schemes with and without alternate relaying, respectively.
    }
    \label{fig:1}
\end{figure}
\begin{figure}[t]
    \centerline{\includegraphics[width=0.57\textwidth]{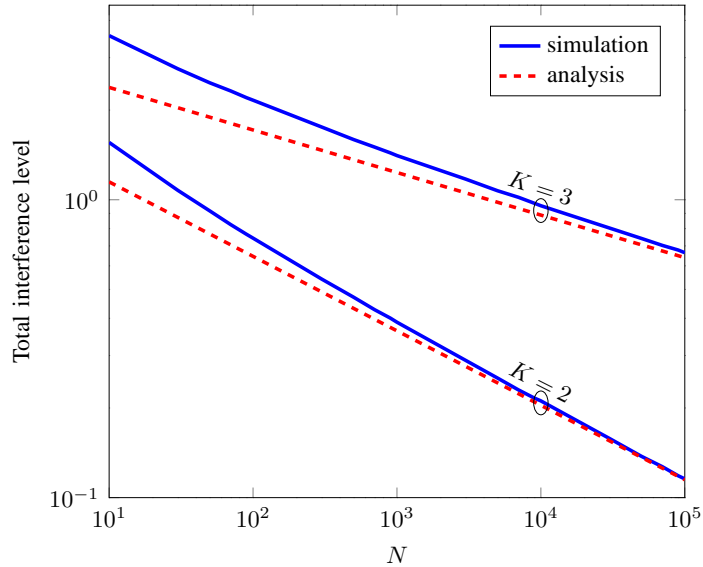}}
    \caption{The average TIL versus $N$ when $K=2,3$ in the $K\times N \times K$ channel.
    }
    \label{fig:2}
\end{figure}
\section{Numerical Evaluation} \label{sec:5}
In this section, we perform computer simulations to validate the
achivability result of the proposed OND scheme in
Section~\ref{sec:3} for finite parameters $N$ and SNR in the
$K\times N\times K$ channel model with interfering relays. In our
simulation, the channel coefficients in (\ref{EQ:y1}) and
(\ref{EQ:y2}) are generated $1\times10^5$ times for each system
parameter.
%\Vgreen{We note that the
%assumption of partitioning the relays into two mutually exclusive
%sets in the proof of Theorem~\ref{thrm:1} is just for the
%tractability purpose; in the practical implementation below, we do
%not need to maintain two separating sets, but rather sequentially
%select relays for the corresponding set according to our protocol
%from the whole relays of the network.}

Figure~\ref{fig:1} shows the achievable sum-rates of the $K\times
N \times K$ channel for the OND schemes with and without alternate
relaying according to $\snr$ in dB scale when $K=2$. Note that $N$
is set to a different scalable value according to $\snr$, i.e.,
$N=\snr^{3K-2}$ for the OND with alternate relaying and
$N=\snr^{2K-2}$ for the OND without alternate relaying,
respectively, to see whether the slope of each curve follows the
DoF in Theorems~\ref{thrm:1} and~\ref{thrm:2}. In the figure, the
dotted green lines are also plotted to indicate the first order
approximation of the achievable rates with a proper bias, where
the slopes are given by $K$ and $K/2$ for the OND schemes with and
without alternate relaying, respectively.

In Fig.~\ref{fig:2}, the log-log plot of the average TIL of the
OND with alternate relaying versus $N$ is shown for the $K\times N
\times K$ channel when $K=2,3$.\footnote{Even if it seems
unrealistic to have a great number of relays in cooperative relay
networks, the range for parameter $N$ is taken into account to
precisely see some trends of curves varying with $N$.} It can be
seen that the TIL tends to decrease linearly with $N$. It is
further seen how many relays are required with the OND scheme with
alternate relaying to guarantee that the TIL is less than a small
constant for a given parameter $K$. In this figure, the dashed
lines are also plotted from theoretical results in
Theorem~\ref{thrm:2} with a proper bias to check the slope of the
TIL. We can see that the TIL decaying rates are consistent with
the relay scaling law condition in Theorem~\ref{thrm:1}. More
specifically, the TIL is reduced as $N$ increases with the slope
of 0.25 for $K=2$ and 0.143 for $K=3$, respectively.

\begin{figure}[t]
    \centerline{\includegraphics[width=0.57\textwidth]{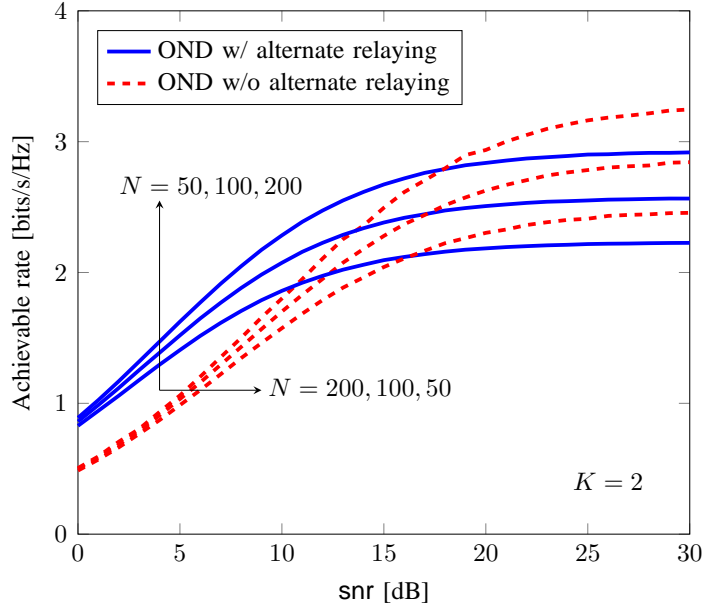}}
    \caption{The achievable sum-rates versus $\snr$ when $K=2$ and $N=100,200$ in the $K\times N \times K$
    channel. Both OND schemes with and without alternate relaying
    are compared.
    }
    \label{fig:3}
\end{figure}

\begin{figure}[t]
    \centerline{\includegraphics[width=0.57\textwidth]{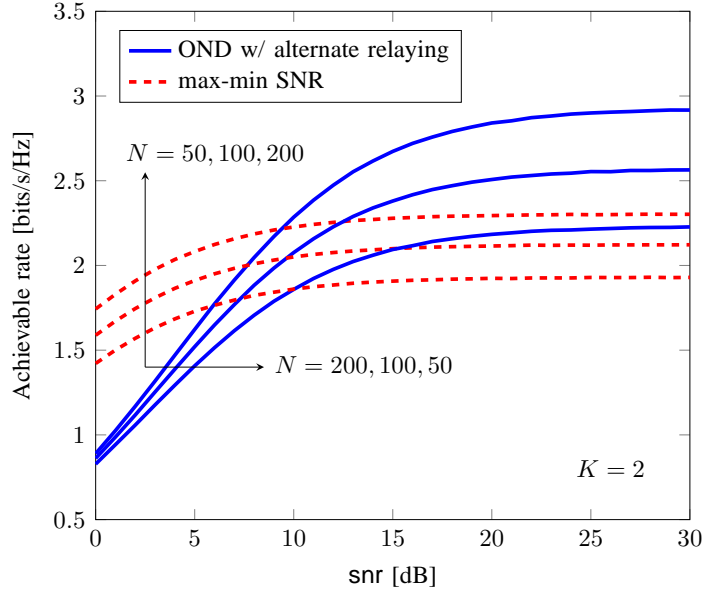}}
    \caption{The achievable sum-rates versus $\snr$ when $K=2$ and $N=100,200$ in the $K\times N \times K$
    channel. The OND scheme with alternate relaying and the
    max-min SNR scheme are compared.
    }
    \label{fig:4}
\end{figure}

%\begin{figure}[t]
%    \centerline{\includegraphics[width=0.44\textwidth]{Figs/Fig6.eps}}
%    \caption{The achievable rates for the OND and max-min SNR schemes of the $K\times N \times K$ channels versus $N$ when $K=2$ and $\snr=10$ dB.
%    }
%    \label{fig:6}
%\end{figure}

Figure~\ref{fig:3} illustrates the achievable sum-rates of the
$K\times N \times K$ channel for the OND schemes with and without
alternate relaying versus $\snr$ (in dB scale) when $K=2$ and
$N=50,100,200$. We can see that in a finite $N$ regime, there
exists the case where the OND without alternate relaying
outperforms that of the OND with alternate relaying. This is
because for finite $N$, the achievable sum-rates for the alternate
relaying case tend to approach a floor with increasing SNR faster
than no alternate relaying case due to more residual interference
in each dimension. We can also see that the crossing points
slightly move to the right as $N$ increases; this is due to the
fact that our OND scheme with alternate relaying always benefits
from having more relays for selection, thus resulting in more
multiuser diversity gain.
% For example, in this figure the crossing points correspond to $\snr=10.5,13,14.5$ for $N=50,100,200$, respectively.
%The effectiveness of the OND scheme can be further ascertained by referring to Fig.~\ref{fig:5} where the achievable rates of the $K\times N \times K$ channels for the alternate and half-duplex relaying schemes are depicted versus $N$ when $K=2$ and $\snr=10,20$ dB. It can be seen that in the low- and moderate-SNR regime, the alternate scheme exhibits a higher rate increase as $N$ increases than the half-duplex scheme whereas the half-duplex scheme dominates in the high-SNR regime. This is because in a high-SNR regime with a finite small value $N$, the interference cannot be sufficiently suppressed in the alternate scheme.
This highly motivates us to operate our system in a switched
fashion when the relay selection scheme is chosen between the OND
schemes with and without alternate relaying depending on the
operating regime of our system.

To further ascertain the efficacy of our scheme, a performance
comparison is performed with a baseline scheduling. Specifically,
in the \emph{max-min SNR} scheme, each S--D pair selects one relay
$\mathcal{R}_i$ ($i\in\{1,\cdots,N\}$) such that the minimum out
of the desired channel gains of two communication links (either
from $\mathcal{S}_k$ to $\mathcal{R}_i$ or from $\mathcal{R}_i$ to
$\mathcal{D}_k$) becomes the maximum among the associated minimum
channel gains over all the unselected relays. This max-min SNR
scheme is well-suited for relay-aided systems if interfering links
are absent. The achievable sum-rates are illustrated in
Fig.~\ref{fig:4} according to $\snr$ (in dB scale) when $K=2$ and
$N=100,200$. We can see that our OND scheme with alternate
relaying outperforms this baseline scheme beyond a certain low SNR
point. We also see that the rate gaps increase when $N$ increases
in the high SNR regime. On the other hand, for fixed $N$, the
sum-rates of the max-min scheme are slightly changed with respect
to $\snr$ due to more residual interference in each dimension.

%\red{Finally, in Fig.~\ref{fig:5}, we validate the performance of the OND scheme with a standardized channel model. Specifically, both large-scale and small-scale fading are taken into account by adopting the 3GPP channel model in \cite{3GPP:10:TS25.467}. The parameters are summarized in Table~\ref{table:4}.  It turns out that the OND scheme still achieves the multiuser gain even with the large-scale fading heterogeneity. }

%%%%%%%%%%%%%%%%%%%%%%%%%%%%%%%%%%%%%%%%%%%%%%%%%%%%%%%%%%%%%%%%%%%%%%%%%%%%%%%%%%%%%%%%%%%%%%%%%%%%%%%%%%%%%%%%%%%%%%%%%%%%%%%%%%%%%%%%%%%%%%%%%%%%%

\section{Concluding Remarks} \label{sec:6}

An efficient distributed OND protocol operating in virtual
full-duplex mode was proposed for the $K\times N\times K$ channel
with interfering relays, referred to as one of multi-source
interfering relay networks. A novel relay scheduling strategy with
alternate half-duplex relaying was presented in two-hop
environments, where a subset of relays is opportunistically
selected in terms of producing the minimum total interference
level, thereby resulting in network decoupling. It was shown that
the OND protocol asymptotically achieves full DoF even in the
presence of inter-relay interference and half-duplex assumption,
provided that the number of relays, $N$, scales faster than
$\snr^{3K-2}$. Numerical evaluation was also shown to verify that
our scheme outperforms the other relay selection methods under
realistic network conditions (e.g., finite $N$ and SNR) with
respect to sum-rates.

Suggestions for future research in this area include the extension
to the MIMO $K\times N\times K$ channel and the optimal design of
joint beamforming and scheduling under the MIMO model.

%%%%%%%%%%%%%%%%%%%%%%%%%%%%%%%%%%%%%%%%%%%%%%%%%%%%%%%%%%%%%%%%%%%%%%%%%%%%%%%%%%%%%%%%%%%%%%%%%%%%%%%%%%%%%%%%%%%%%%%%%%%%%%%%%%%%%%%%%%%%%%%%%%%%%

\end{document}